\documentclass[12pt]{amsart}

\usepackage[utf8]{inputenc}
\usepackage{tikz}
\usetikzlibrary{shapes,arrows}
\usepackage{amsaddr} 
\usepackage{amsthm}
\usepackage{amsmath,amsfonts,amssymb}
\usepackage[letterpaper]{geometry}
\usepackage{color} 
\usepackage{graphicx, xcolor}
\usepackage{subcaption,tabularx,float}
\usepackage{hyperref}
\usepackage[capitalise]{cleveref}
\crefformat{thm}{Thm.~#2#1#3}
\crefformat{lemma}{Lem.~#2#1#3}
\crefformat{definition}{Defn.~#2#1#3}
\crefformat{equation}{(#2#1#3}
\crefformat{appendix}{App.~#2#1#3}

\newcommand{\set}[1]{\left\lbrace #1 \right\rbrace}

\newcommand{\C}{\mathcal{C}}

\newcommand{\R}{\mathbb{R}}
\newcommand{\T}{\mathbb{T}}

\newcommand{\iotab}{{\iota\hspace{-0.6em}\raisebox{0.75pt}{\,--}}}

\renewcommand{\L}{\mathcal{L}}

\newcommand{\U}{\mathcal{U}}

\DeclareMathOperator{\crit}{Crit}
\DeclareMathOperator{\Div}{div}
\DeclareMathOperator{\Vol}{Vol}

\theoremstyle{plain}
\newtheorem{thm}{Theorem}[section]

\newtheorem{corollary}[thm]{Corollary}
\newtheorem{lemma}[thm]{Lemma}

\theoremstyle{definition}
\newtheorem{definition}[thm]{Definition}

\newtheorem{remark}[thm]{Remark}

\begin{document}
\title{Integrability, Normal Forms and Magnetic Axis Coordinates}
\author{J. W. Burby$^1$, N. Duignan$^2$, and J. D. Meiss$^2$}
\address[foot]{(1) Los Alamos National Laboratory, Los Alamos, NM 97545 USA\\ 
                (2) Department of Applied Mathematics, University of Colorado, Boulder, CO 80309-0526, USA}


\begin{abstract}
Integrable or near-integrable magnetic fields are prominent in the design of plasma confinement devices. 
Such a field is characterized by the existence of a singular foliation  consisting entirely of invariant submanifolds. A regular leaf, known as a flux surface,of this foliation must be diffeomorphic to the two-torus. In a neighborhood of a flux surface, it is known that the magnetic field admits several exact, smooth normal forms in which the field lines are straight. However, these normal forms break down near singular leaves including elliptic and hyperbolic magnetic axes. In this paper, the existence of exact, smooth normal forms for integrable magnetic fields near elliptic and hyperbolic magnetic axes is established. In the elliptic case, smooth near-axis Hamada and Boozer coordinates are defined and constructed. Ultimately, these results establish previously conjectured smoothness properties for smooth solutions of the magnetohydrodynamic equilibrium equations. The key arguments are a consequence of a geometric reframing of integrability and magnetic fields; that they are presymplectic systems.
\end{abstract}

\date{\today}

\maketitle
\tableofcontents

\section{Introduction}\label{sec:Intro}
Toroidal Magnetic confinement devices, such as tokamaks and stellarators, are designed to confine particles with a strong magnetic field $B$ that lies, as much as possible, on a set of nested tori \cite{Hazeltine03, Helander14}. In magnetohydrodynamic equilibria, i.e., for Magneto-Hydro-Statics (MHS), these surfaces are iso-pressure surfaces. As is well-known, the incompressiblity of the field $B$ makes the ``dynamics'' of the field-line flow similar to that of a Hamiltonian system. This fact together with the invariance of iso-pressure surfaces implies that, in some sense, the system should be integrable.

We will say that a non-vanishing magnetic field $B$ is \emph{integrable} if there is a divergence-free vector field $J$ and a function $p$ such that
\begin{equation}\label{eq:commuting-def}
\begin{split}
    [B,J] &\equiv B \cdot \nabla J - J \cdot \nabla B =0 \\
    J\times B &= \nabla p
\end{split}
\end{equation}
where $\nabla p \neq 0 $ \textit{ almost everywhere}. The first condition states 
that $J$ commutes with $B$ and so is a \textit{symmetry} vector field of $B$. By contrast, the second condition in \eqref{eq:commuting-def} states that that $p$ is an invariant of both $B$ and $J$, since, $B \cdot \nabla p = 0$ and $J \cdot \nabla p = 0$. Indeed, in \S\ref{sec:Perspective} we think of $p$ as a \textit{Hamiltonian} and will call the pair $(J,p)$ a \emph{Hamiltonian pair} associated with $B$. Note that \eqref{eq:commuting-def} implies that $\nabla\times(J\times B) = 0$ since $B$ and $J$ are both divergence-free.  

The conditions \eqref{eq:commuting-def} are independent of the physical interpretation of $J$ as the current. For MHS, the condition $J = \nabla \times B$ must be added, and $p$ interpreted as the hydrostatic pressure. With these two additional conditions, we say that the magnetic field is \emph{MHS integrable}. 

Since both $J$ and $B$ are tangent to surfaces of constant $p$, the phase portrait of an integrable magnetic field is characterized by the topology of the level sets of $p$. A closed surface along which $p=\text{constant}$ and $\nabla p \neq 0$ is known as a \emph{regular flux surface}. By contrast, a \emph{magnetic axis} is a closed $B$-line on which $\nabla p = 0$.

The structure of an integrable $B$ in a neighborhood of a compact, regular flux surface has been well-understood for decades. Assuming MHS integrability, Hamada \cite{Hamada62} showed that, in such a neighborhood, there are special coordinates $(\psi,\theta,\zeta)$ where $\psi$ is a radial or flux variable, and $\theta, \zeta$ are angular variables. In these coordinates, $p$ depends only on the flux, $\psi$, and the magnetic field and the vector field $J$ have only two contravariant components that are also flux functions. For example,
\[
    B^\theta(\psi) = B \cdot \nabla  \theta, \quad 
    B^\zeta(\psi) = B \cdot \nabla \zeta, 
\]
and $B^\psi = 0$ (see \S\ref{sec:Summary} for more details).
The implication is that the field lines of $B$ and $J$ are straight in these coordinates. We remark in \S\ref{sec:Applications} that existence of Hamada coordinates does not require MHS integrability; our weaker notion of integrability is sufficient.

Assuming MHS integrability, Boozer \cite{Boozer81} showed there are flux coordinates, which are constructed so that the field lines of $B$ and $B\times \nabla p$ are straight, or, equivalently, so that the angular covariant components of $B$ are flux functions.
Note that $B\times\nabla p$ is proportional to the perpendicular (or ``diamagnetic'') current:
\[
   B\times \nabla p   = B \times (J \times B)
   = |B|^2\left( J - B \frac{J \cdot B}{|B|^2}\right)
   = |B|^2 J_\perp .
\]
In contrast to the Hamada case, our generalized notion of integrability is not sufficient to imply the existence of Boozer coordinates. 
Instead, MHS integrability is a sufficient condition for these coordinates to exist.

There is a strong analogy between these coordinate systems and those that can be found for an integrable Hamiltonian system. Indeed, in such a system, near any compact, regular level set of the
invariants, there are action-angle variables: this is the Liouvllle-Arnol'd theorem \cite{Arnold78}. For the magnetic field case, $B$ plays the role of the
Hamiltonian vector field, but, as we will argue in \S\ref{sec:Perspective}, it has a Hamiltonian that is effectively zero. The field $J$ plays the role of a symmetry vector field and $p$ is the corresponding invariant. In this case, Hamada coordinates correspond to action-angle variables. Boozer coordinates reflect the alternative choice of $B/|B|^2$ as the Hamiltonian vector field and  $(B \times \nabla p)/|B|^2$ as the symmetry vector field.

Neither the Hamada nor the Boozer normal form is valid in the neighborhood of a magnetic axis since the construction of each of these coordinate systems requires $\nabla p \neq 0$. However, since field line dynamics near a magnetic axis may be written as a non-autonomous Hamiltonian system, the near-axis magnetic field must admit an asymptotic Birkhoff normal form \cite{Arnold78}. The latter result, while not completely satisfactory since Hamada and Boozer coordinates are smooth rather than merely formal, does suggest that a smooth normal form theory near a magnetic axis may exist. 

The goal of this paper is to investigate the consequences of integrability in the neighborhood of a \emph{nondegenerate} magnetic axis, i.e. a magnetic axis along which the matrix of second derivatives normal to the axis is non-singular. We will establish in \S\ref{sec:NormalForm} the existence $C^\infty$ and analytic near-axis coordinates that, as we see in \S\ref{sec:Applications} reduce to a normal forms analogous to Hamada's coordinates near the axis. We will also show in \S\ref{sec:Applications} that for the special case of MHS integrability, there exist $C^\infty$ near-axis Boozer coordinates.

\section{Summary of the Main Results} \label{sec:Summary}

In this paper it will be assumed that a \emph{magnetic field} $B$ is a non-vanishing, $C^\infty$, divergence-free vector field defined on a smooth manifold $M$ with volume element $\Omega$. As a special case, we may take $M=U$, for some open region $U$ in $\mathbb{R}^3$ equipped with the usual volume element $\Omega = dx\,dy\,dz$.
While naturally-occurring magnetic fields may have nulls, those used in the most promising candidate designs for magnetically-confined fusion reactors typically do not \cite{Hazeltine03,Helander14}. Our results therefore apply to the common magnetic fields used for plasma confinement.

Much is known about the structure of integrable magnetic fields. Recall that our definition in \eqref{eq:commuting-def} of integrability does not require MHS. 
As we will more formally recall in \S\ref{sec:NormalForm}, if $\Lambda\subset U$ 
is a compact $C^\infty$ surface on which $p$ is constant and $\nabla p$ is nowhere vanishing, then $\Lambda$ must be a two-torus. Moreover, as we formally state in \cref{thm:RegularHamada}, there are $C^\infty$ coordinates $(\psi,\theta,\zeta)$ defined in a (possibly narrow) tubular neighborhood of $\Lambda$ such that 
\begin{equation}\label{eq:HamadaCoord}
\begin{split}
    p & = p(\psi) ,\\
    B &= F^\prime(\psi)\,\nabla\psi\times \nabla\theta - G^\prime(\psi)\,\nabla\psi\times \nabla\zeta,\\
    J &= K^\prime(\psi)\,\nabla\psi\times \nabla\theta-L^\prime(\psi)\,\nabla\psi\times \nabla\zeta,
\end{split}
\end{equation}
where $p,F,G,K,L$ are $C^\infty$ functions of $\psi$, and $\theta,\zeta$ are $2\pi$-periodic angular variables \cite{Hamada62, Hazeltine03, Helander14}. Such coordinates are known as \emph{Hamada coordinates}. According to \eqref{eq:HamadaCoord}, $B$ and $J$ in Hamada coordinates are placed in a \emph{normal form}. In particular, each of $B,J$ has only two contravariant components. 

If $B$ also happens to be MHS integrable then, as we formally state in \cref{thm:RegularHamada}, there are $C^\infty$ coordinates $(\psi,\theta,\zeta)$ defined in a tubular neighborhood of $\Lambda$ such that
\begin{equation}\label{eq:BoozerCoord}
\begin{split}
    p& = p(\psi) ,\\
    B & = f^\prime(\psi)\,\nabla\psi\times \nabla\theta - g^\prime(\psi)\,\nabla\psi\times \nabla\zeta ,\\
    B\times\nabla p & = k^\prime(\psi)\,\nabla\psi\times \nabla\theta - l^\prime(\psi)\,\nabla\psi\times \nabla\zeta,
\end{split}
\end{equation}
where $p,f,g,k,l$ are $C^\infty$ functions of $\psi$ and $\theta,\zeta$ are $2\pi$-periodic angular variables. Such coordinates are known as \emph{Boozer coordinates} \cite{Boozer81, Helander14}. 
In this case $B$ and $B\times\nabla p$ are placed in normal form. This contrasts with the Hamada form \eqref{eq:HamadaCoord}. If $B$ is expressed in the covariant representation
\[
    B = B_\psi \nabla \psi + B_\theta \nabla \theta + B_\zeta \nabla \zeta ,
\]
then the angular covariant components are flux functions: $k'(\psi) = -B_\theta p'(\psi)$
and $l'(\psi) = B_\zeta p'(\psi)$. Imposing the latter property was how Boozer originally defined his eponymous straight-line coordinates \cite{Boozer81}.

We give a precise summary of the existence theory for Hamada and Boozer coordinates in \S\ref{sec:Applications}.

Hamada and Boozer coordinates play foundational roles in many theoretical analyses of magnetic confinement devices. However, these coordinate systems have an important flaw; they are ill-defined on any \emph{magnetic axis}, i.e., a closed $B$-line along which $\nabla p = 0$.  In fact, neither Hamada coordinates nor Boozer coordinates may be continuously defined in an open region that includes a magnetic axis, because each level set of $p$ contained in the domain of either type of coordinate system must be diffeomorphic to a two-torus. Therefore, any possible normal form for the vector fields in a neighborhood of a magnetic axis must differ from \eqref{eq:HamadaCoord}-\eqref{eq:BoozerCoord}.

Nevertheless, one might hope that the behavior of these coordinate systems near an axis is no more complicated than standard polar coordinates $(r,\theta)$ on $\mathbb{R}^2$ near the origin. Indeed, though the domain of the coordinates $(r,\theta)$ is the punctured plane $\mathbb{R}^2_0 = \{(x,y)\mid (x,y)\neq (0,0)\}$, by defining Cartesian coordinates, $(r,\theta)\mapsto (r\cos\theta,r\sin\theta)$, one obtains a coordinate system that is merely the restriction to $\mathbb{R}^2_0$ of smooth coordinates on all of $\mathbb{R}^2$. In other words, these functions have a \emph{removable} singularity at $r=0$. So if $(\psi,\theta,\zeta)$ comprises a system of, say, Hamada coordinates with $\psi = 0$ corresponding to the magnetic axis, one might hope that the coordinates $(X,Y,\zeta) = (\sqrt{2\psi}\cos\theta,\sqrt{2\psi}\sin\theta,\zeta)$ 
similarly exhibit only a removable singularity at $\psi=0$.

However, without further insight into the nature of Boozer and Hamada coordinates near a magnetic axis, there is no obvious justification for such an assumption. For a simple counterexample, instead of standard polar coordinates, suppose one were to define the new functions $(R,\Theta) = (r,\theta + r^{-1})$, which also comprise a valid coordinate system on $\mathbb{R}^2_0$. However, the corresponding ``Cartesian" coordinates $(X,Y)=(R\cos\Theta,R\sin\Theta)$ are generally ill-behaved since, for example,
\[
    \nabla X  =(\cos\Theta + r^{-1}\sin\Theta) \nabla r -r\sin\Theta \nabla\theta 
\]
 diverges as $r\rightarrow 0$.

We will address the case that the magnetic axis is \textit{nondegenerate}. An \emph{elliptic magnetic axis} is a closed field line along which $\nabla p = 0$, and $p$ is a normal
minimum or maximum, i.e., the normal Hessian, $D^2_\perp p$, is nondegenerate with signature $(+,+)$ or $(-,-)$. Similarly, a \emph{hyperbolic magnetic axis} is a normal saddle, so that the normal Hessian has signature $(+,-)$. If local (un)stable manifolds of the saddle are non-orientable, corresponding to a M\"{o}bius strip, then the axis is said to be \emph{reflection hyperbolic}; if the manifolds are orientable, an annulus, then it is said to be \emph{direct hyperbolic}.
More details on the various topologies are given in \S\ref{sec:Perspective}.

In \S\ref{sec:NormalForm} we establish the following basic existence result in a neighborhood of a nondegenerate magnetic axis. 
\begin{thm}\label{sec2_thm1}
Suppose that $B$ is an integrable magnetic field with Hamiltonian pair $(J,p)$, and $\gamma$ is a nondegenerate magnetic axis, all $C^\infty$ (or analytic). There are coordinates $(x,y,\phi)$ in a covering tubular neighborhood $\U$ of $\gamma$ and
$C^\infty$ coordinates $(x,y,\phi)$ such that $\gamma = \{(0,0,\phi):\phi \in \mathbb{R}\,\mathrm{ mod }\,2\pi\}$, and
\begin{align*}
     p& = P(\psi)\\
     B& = \nabla y\times \nabla x 
        - \Psi^\prime(\psi)\,\nabla\psi\times \nabla\phi,
\end{align*}
where $P$ and $\Psi$ are $C^\infty$ (or analytic) and $\psi =  \tfrac12(x^2 \pm y^2)$ in the elliptic/hyperbolic cases, respectively. Moreover, in the reflection hyperbolic case these coordinates give a double cover of $\U$.
\end{thm}
\noindent Though this theorem applies both to elliptic and hyperbolic axes, its proof relies on very different arguments in the two cases. We may describe the essential difficulties encountered in the proofs in broad strokes as follows.

The elliptic case is simplest. After moving into coordinates $(X,Y,\Phi)$ in which $p$ is  a function only of $\tfrac12(X^2+Y^2)$, the crux of the proof in \S\ref{sec:ellipticproof} is establishing existence of a smooth solution $f: U \to \R$ of the ``magnetic differential equation''
\begin{align}
    B_\lambda\cdot\nabla f = B_\lambda\cdot a,\label{mdiff}
\end{align}
where $B_\lambda = \nabla\times A_\lambda$, $\lambda\in [0,1]$, is a divergence-free vector field that linearly interpolates between the ``old" magnetic field $B_0$ and the ``new" normal form magnetic field $B_1$, and $a = \partial_\lambda A_\lambda$. As is well-known in plasma physics, such magnetic differential equations are generally plagued by singular divisor problems. Equation \,\eqref{mdiff} also enjoys the additional complication of requiring the solution $f$ to be independent of $\lambda$ even though the coefficients are $\lambda$-dependent. To sidestep these difficulties, we show that the toroidal and poloidal fluxes for $B_\lambda$ are indepedent of $\lambda$. In addition, we will see that the interpolating field $B_\lambda$ is tangent to $\psi$-surfaces. These two facts imply the existence of the single-valued function $f$ that satisfies \eqref{mdiff}.

The hyperbolic case, discussed in \S\ref{sec:hyperbolicproof}, is more complicated because the level sets of $p$ near the axis are not necessarily connected. Indeed, the approach used in the elliptic case to establish that $p$ and $\Psi$ are smooth functions of $\psi$ no longer works. In this case we  will leverage preexisting results for 4D Hamiltonian systems by first embedding the 3D magnetic system in 4D. Under this \textit{coisotropic} embedding, $B$ extends to a Hamiltonian vector field on the larger space. In fact, we show that $J$ can also be embedded in such a way that the extended $B$ is integrable in the 4D space. Ultimately, this allows for the immediate application of the existing results on Hamiltonian systems. The desired theorem then emerges through projection onto the embedded 3D magnetic system.

For the elliptic magnetic axis, we will then build upon \cref{sec2_thm1} to construct $C^\infty$ Hamada and Boozer coordinates near a nondegenerate elliptic magnetic axis in \S\ref{sec:Applications}.

\begin{thm}[Near Axis Hamada Coordinates] \label{sec2_NAH}
Suppose that $B$ is an integrable magnetic field with Hamiltonian pair $(J,p)$ and elliptic magnetic axis $\gamma$. There is an open set, $\U$, containing $\gamma$ and $C^\infty$ coordinates $(x,y,\phi)$ defined on $\U$ such that $\gamma = \{(0,0,\phi): \phi \in \mathbb{R}\,\mathrm{ mod }\,2\pi\}$ 
and
\begin{equation}\label{eq:NAH}
\begin{split}
    p& = p(\psi)\\
    B &= F^\prime(\psi)\,\nabla y\times\nabla x - G^\prime(\psi )\,\nabla\psi \times \nabla\phi \\
    J& = K^\prime(\psi )\,\nabla y\times\nabla x -L^\prime(\psi )\,\nabla\psi \times \nabla\phi ,
\end{split}
\end{equation}
where $\psi  = \tfrac12(x ^2 + y ^2)$, and all the functions are $C^\infty$.
\end{thm}

The construction of  these coordinates in \S\ref{sec:Applications} begins by using the coordinates $(x,y,\phi)$ provided by \cref{sec2_thm1}, and proceeds by appropriately shifting the coordinate system along the field lines of $B$. This is similar in spirit to the construction of Hamada coordinates near a regular invariant torus \eqref{eq:HamadaCoord}, but requires special care in order to establish smoothness on the magnetic axis.

In these coordinates, the volume of a region $R$ is given by $\Vol(R)=\int_R\,\rho \,dx \,dy \,d\phi $, where $\rho$ is a smooth density. Using \eqref{eq:commuting-def} and the normal forms \eqref{eq:NAH}, it is not difficult to show that $\rho$ must be a smooth function of $\psi$. In other words, one would say that the Jacobian is a flux function in near-axis Hamada coordinates.

While the near-axis Hamada coordinates only require that $B$ is integrable, near-axis Boozer coordinates require that $B$ is MHS integrable.
 
\begin{thm}[Near Axis Boozer Coordinates] \label{sec2_NAB}
Suppose $B$ is an MHS integrable magnetic field with Hamiltonian pair $(J=\nabla\times B,p)$ and elliptic magnetic axis $\gamma$. Then there is an open set $\U$ containing $\gamma$ and $C^\infty$ coordinates $(x ,y ,\phi )$ defined on $\U$ such that
\begin{equation}\label{eq:NAB}
\begin{split}
    p& = p (\psi )\\
    B &= f^\prime(\psi )\,\nabla y \times\nabla x  - g^\prime(\psi )\,\nabla\psi \times \nabla\phi \\
    B\times \nabla p& = k^\prime(\psi )\,\nabla y \times\nabla x -l^\prime(\psi )\,\nabla\psi \times \nabla\phi ,
\end{split}
\end{equation}
where $\psi  = \tfrac12(x ^2 + y ^2)$,  and all the functions are $C^\infty$ single-variable functions.
\end{thm}

In order to find the coordinates \eqref{eq:NAB}, we first observe that the vector fields
\begin{equation}\label{eq:Bprime}
    B^\prime = \frac{B}{|B|^2},\quad J^\prime  = \frac{B \times \nabla p}{|B|^2}
\end{equation}
satisfy
\begin{gather*}
    \nabla\cdot(|B|^2 B^\prime)  =\nabla\cdot(|B|^2 J^\prime)  = 0 \\
     [J',B'] = 0,\qquad |B|^2(J^\prime\times B^\prime)  = \nabla p,
\end{gather*}
assuming that $J=\nabla\times B$. These are exactly the conditions for $B^\prime$ to be an integrable magnetic field with Hamiltonian pair $(J^\prime,p)$ if volumes were computed according to $\Vol(U) = \int_U\,|B|^2\,d^3x$ instead of $\Vol(U) = \int_U \,d^3x$. Because we formulate our existence result for near-axis Hamada coordinates on 3D manifolds with arbitrary volume elements, existence of near-axis Boozer coordinates then follows immediately from the corresponding result for Hamada coordinates. In other words, our argument reveals that existence of near-axis Boozer coordinates is merely a corollary of existence of near-axis Hamada coordinates. Note that the converse statement cannot be true in general because Hamada coordinates only require an integrable magnetic field, while Boozer coordinates require MHS integrability.

If $(x,y,\phi)$ is a system of near-axis Hamada (Boozer) coordinates, then away from the axis the coordinates $(\psi,\theta,\zeta)$ defined through $x=\sqrt{2\psi}\cos\theta$, $y=\sqrt{2\psi}\sin\theta$, $\zeta = \phi$ comprise Hamada (Boozer) coordinates in the usual sense.

An immediate application of \cref{sec2_NAH} is to fill the gap in the proof of Theorem 1 from \cite{Burby_Kallinikos_MacKay_2020} concerning the existence and smoothness of a volume-preserving $\T$-symmetry for an MHS integrable magnetic field. An implication of \cref{sec2_NAH} is

\begin{corollary}\label{sec2_corollary}
If $B$ is a integrable magnetic field with Hamiltonian pair $(J,p)$ and  elliptic magnetic axis $\gamma$, and all functions are $C^\infty$, then there is an open set $U$ containing $\gamma$ and a $C^\infty$ volume-preserving $\T$ action $T_\zeta:U\rightarrow U$, $\zeta\in \T$, with nowhere-vanishing infinitesimal generator $u$ that satisfies $[u,B] = [u,J] = 0$.
\end{corollary}
\begin{proof}
In near-axis Hamada coordinates, the desired $\T$ action is just $T_\zeta(x ,y ,\phi ) = (x ,y ,\phi +\zeta) $ and the nowhere-vanishing infinitesimal generator is $u=\partial_\zeta$. $T_\zeta$ is volume-preserving and $u$ commutes with $B$ and $J$ because the Jacobian in near-axis Hamada coordinates is a flux function.
\end{proof}

\noindent The argument presented in the proof of Theorem 1 in \cite{Burby_Kallinikos_MacKay_2020} implies the existence of a $C^\infty$ volume-preserving $\T$ action that commutes with $B$ and $J$ in a neighborhood of any regular $p$-surface. Since both the $\T$ action given by  \cref{sec2_corollary} and the $\T$ action given in \cite{Burby_Kallinikos_MacKay_2020} are equivalent to translation along $\zeta $ in Hamada coordinates, the preceding remarks serve to complete the proof of Theorem 1 from \cite{Burby_Kallinikos_MacKay_2020} \emph{provided} the two $\T$ actions can be smoothly glued together. That such gluing can be done follows from the following uniqueness property satisfied by Hamada coordinates: if $(\psi,\theta,\zeta)$ and $(\overline{\theta},\overline{\zeta},\psi)$ are two systems of Hamada coordinates such that the periods of $\nabla\theta,\nabla\zeta$ agree with those of $\nabla\overline{\theta},\nabla\overline{\zeta}$ then $\overline{\theta} = \theta + c_1(\psi)$ and $\overline{\zeta} = \zeta + c_2(\psi)$, where $c_1,c_2$ are $C^\infty$ single-variable functions. For a detailed characterization of Hamada coordinates, see \cref{hamada_characterized}.

\cref{sec2_NAB} on the existence of $C^\infty$ near-axis Boozer coordinates may also be used to fill logical gaps in previous work on near-axis expansions of MHS integrable fields. For example, Section II.C of \cite{Garren_Boozer_91a} (c.f. Eqs.\,(30)-(43) in that reference) asserts without proof the existence of analytic near-axis Boozer coordinates assuming an analytic MHS integrable magnetic field. For the reason mentioned above, the usual construction of Boozer coordinates does not obviously imply the existence of near-axis Boozer coordinates with any particular regularity. However, since analytic  MHS integrable fields are infinitely differentiable, \cref{sec2_NAB} does imply the existence of $C^\infty$ near-axis Boozer coordinates assuming analytic (and more generally $C^\infty$) MHS integrability. Continuous differentiability is more than enough to justify the formal power series expansions studied in \cite{Garren_Boozer_91a}. These remarks apply without modification to the near-axis expansion introduced in Section II.C of \cite{Garren_Boozer_91b}. \cref{sec2_NAB} may also be used to simplify and extend the low-order perturbative construction of near-axis Boozer coordinates given in \cite{Landerman_Sengupta_2018} (cf Appendix A in that reference) and applied in the series of papers \cite{Landerman_Sengupta_2018,Landerman_Sengupta_Plunk_2019,Plunk_Landerman_Helander_2019,Landerman_Sungupta_2019}. Our result immediately implies that near-axis Boozer coordinates exist to all orders in perturbation theory assuming an analytic or $C^\infty$ MHS integrable magnetic field.

\section{Field-Line Flow as a Presymplectic System}\label{sec:Presymplectic}

The geometry of field-line flow is not intrinsically Hamiltonian since the phase space is three-dimensional. It is often assumed that if the field is everywhere non-vanishing one can use, for example, the toroidal coordinate as an effective time and think of the field-line flow as that of a nonautonomous Hamiltonian. In a properly chosen coordinate system, this gives a symplectic dynamical system with the canonically conjugate coordinates given by the toroidal flux and poloidal angle \cite{Boozer81}.

Nevertheless, we assert that this is not the natural geometry in which to consider this problem.
In this paper we argue that a more natural foundation is that of a \textit{presymplectic} form.
In particular,  the presymplectic formulation of field-line flow requires no assumptions on $B$ other than it be divergence-free.

\subsection{A perspective on the geometry of field-line flow}\label{sec:Perspective}
A presymplectic form on a manifold $M$ is analogous to the canonical two-form, $\omega$; 
the nondegenerate, closed two-form, of symplectic geometry, see \cref{sec:Weak}.
Recall, for a Hamiltonian $H$, the flow vector field $X_H \in \mathfrak{X}(M)$ is generated by the equations
\begin{equation}\label{eq:Hamilton}
    \iota_{X_H}\omega = -dH. 
\end{equation}
In canonical coordinates $(q,p)$, where $\omega = dp\wedge dq$, one obtains
\[
    X_H = (\dot q, \dot p) = (\partial H/\partial p, -\partial H /\partial q).
\]

In lieu of the canonical formalism, we begin with a two-form $\beta$ on a three-dimensional space that will correspond to the ``flux form'' of the magnetic field $B$, i.e., for any surface $S$,
\[
    \mathcal{F}_S = \int_S \beta = \int B \cdot \hat{n} d^2S
\]
is the flux of $B$ through $S$ with unit normal $\hat{n}$ and surface area element $d^2 S$. Indeed, we
can reformulate this integral in a coordinate-free way in terms the Riemannian volume form $\Omega$
\[
    \mathcal{F}_S = \int_S \iota_B\Omega.
\]
In general, a vector field $B \in \mathfrak{X}(M)$ can be thought of as an operator on functions, e.g., in coordinates $(x^1,x^2,x^3)$,
\[
    B = B^i \frac{\partial}{\partial x^i},
\]
where $B^i$ is the $i^{th}$ contravariant component of $B$ and we use the summation convention.
In the same coordinates, the Riemannian volume form may be written $\Omega = \rho dx^1 \wedge dx^2 \wedge dx^3$, where $\rho = (\nabla x^1 \cdot \nabla x^2 \times \nabla x^3)^{-1}$. Therefore
\begin{equation}\label{eq:EuclideanBeta}
   \iota_B\Omega = \rho \left(B^1 dx^2\wedge dx^3+ B^2 dx^3\wedge dx^1+ B^3 dx^1 \wedge dx^2\right).
\end{equation}
The equality of the two expressions for $\mathcal{F}_S$ then implies that $\beta = \iota_B \Omega$. This two-form is the natural starting point from which to investigate the field-line flow of $B$. 

The form $\beta$ is called a \textit{presymplectic} form. It is a closed two-form, 
like the canonical form $\omega$, but it is necessarily degenerate.
Indeed, inserting the field $B$ into the flux form \eqref{eq:EuclideanBeta} gives zero:
\begin{align*}
    \iota_B\beta &= \rho \left[(B^2B^3-B^3B^2)dx^1 + (B^3B^1-B^1B^3)dx^2 +(B^1B^2-B^2B^1)dx^3\right] \\
             &= (B \times B)^\flat  = 0 .
\end{align*}
Here ${}^\flat$ denotes the index-lowering operator associated with a metric $g$, so that 
for any vector field $v$,
\begin{equation}\label{eq:flatDef}
    v^\flat = \iota_v g = v_i dx^i
\end{equation}
is the associated one-form, where $v_i$ are the covariant components, $v_i = g_{ij} v^j$, see e.g., \cite{MacKay20}.
Whenever the dimension of the space is odd, the two-form $\beta$ must be degenerate; this is to be contrasted with the assumed nondegeneracy of a symplectic form (which requires an even-dimensional space). The equation $\iota_B \beta = 0$ is analogous to Hamilton's equations \eqref{eq:Hamilton}, though here the Hamiltonian is effectively zero or a constant.

Recall that an $n$ degree-of-freedom Hamiltonian system is (Liouville-Arnold) integrable if it has a set of $n$ independent, Poisson-commuting invariants $\{F_1,F_2,\ldots,F_n\}$ \cite{Arnold78}. Due to the nondegeneracy of the symplectic form $\omega$, each invariant,  when thought of as a Hamiltonian, gives rise, as in \eqref{eq:Hamilton}, to a  vector field through $\iota_{X_{F_i}}\omega = -dF_i$. If $F_i$ is an integral, the vector field $X_{F_i}$ is a symmetry of $X_H$ in that the two vector fields commute, $[X_H,X_{F_i}] = 0$. 

Since the three-dimensional system $\iota_B \beta = 0$ is like  a $1\tfrac12$ 
degree-of-freedom Hamiltonian system, one should expect that a single nontrivial invariant
should be sufficient for integrability. 
In the plasma context, the invariant corresponds to a scalar pressure $p$; this
is invariant under the field-line flow if $B$ lies in surfaces of constant $p$, i.e., 
$B \cdot \nabla p = 0$, or equivalently $\ \L_B p = 0$ for the Lie derivative $\L$. 
If we denote the Hamiltonian vector field associated with this invariant by  $J$, it must solve $\iota_{J}\beta = -dp$, which becomes
\begin{equation}\label{eq:MHS}
   -dp = \iota_J\beta = \iota_J\iota_B \Omega = (B \times J)^\flat \quad 
   \Rightarrow \quad  J \times B = \nabla p .
\end{equation}
As a consequence of the degeneracy of $\beta$, \eqref{eq:MHS} only determines $J$ up to an arbitrary component parallel to $B$.

Equation \eqref{eq:MHS} implies that $J$ must also lie in surfaces of constant $p$. In MHS, we take $J = \nabla \times B$ so that $J$ is the current. Note that 
this then implies that $\nabla \cdot J = 0$, so $J$ is also a divergence free vector field. 
More generally, whenever $J$ is divergence free, it is the required symmetry of $B$ because the commutator then vanishes:
\begin{align*}
    [B,J] &= (B\cdot \nabla) J -(J\cdot \nabla) B \\
          &= \nabla \times (J \times B) - J(\nabla \cdot B) + B(\nabla \cdot J)\\
          &= \nabla \times \nabla p  = 0 ,
\end{align*}
since both $B$ and $J$ are divergence free.

Reformulating \eqref{eq:commuting-def} in terms of $\beta$, we will say that a magnetic field $B$ is \emph{integrable} if there is a $C^\infty$ divergence-free vector field $J$ and a $C^\infty$ function $p$ such that
\[
    \iota_J\beta = -dp
\]
A magnetic field that arises as a solution of the magnetohydrostatic (MHS) equation, \eqref{eq:MHS}, is integrable with $J=\nabla \times B$ and $p$ equal to the hydrostatic pressure. We emphasize, however, that most of the results in this paper do not require that $J$ be the current and $p$ the pressure. Whenever a result applies only to MHS integrable magnetic fields, we will clearly indicate that this is the case.

In the following sections, we investigate the consequences of integrability in the neighborhood of a magnetic axis. For the Liouville-Arnold case, it is well-known that the 
intersection of levels sets of the $n$ invariants is an $n$-torus if it is compact and the invariants have independent gradients. One says that integrable systems are foliated by such invariant tori. On each invariant torus, the commuting symmetries from the invariants give rise to a ``torus-action'' and the flow of
the Hamiltonian is conjugate to a rigid rotation with some frequency vector. For the magnetic field case,
the analog is simply that the iso-pressure surfaces are tori and $B$ has a well-defined rotational transform, $\iotab(p)$, on each surface.

However, this structure breaks down near a degenerate level set, and a magnetic axis is such a degeneracy.
For MHS the standard case corresponds 
to the center of the plasma where the pressure is maximum.
However, any closed field line along which $\nabla p = 0$ corresponds to such an ``axis.''
Even for integrable Hamiltonian systems, the dynamics near such a degenerate level set can be more complicated. A standard example is the separatrix of the integrable pendulum. For the field case, such hyperbolic magnetic field lines occur whenever there are magnetic islands, and these generically occur for rational rotational transform.

\subsection{Presymplectic forms and Hamiltonian flows} \label{sec:presymplecticForms}

Generally, a \textit{presymplectic} form is a closed two-form on a manifold $M$. 
When this form is nondegenerate, it is said to be \textit{symplectic} \cite{Ortega04}. We recall this and several other related concepts in \cref{sec:Weak}.

We will specialize here to the case that $M$ is an orientable three-manifold, as is
appropriate for magnetic fields. As we remarked in \S\ref{sec:Perspective}, every two-form on an odd-dimensional manifold must be degenerate. 
We will assume that there is a closed, presymplectic two-form $\beta$ 
of maximal rank on $M$. Any two-form $\beta$ induces a linear map, $\hat{\beta}$, from vector fields, $X \in \mathfrak{X}(M)$, to one-forms, $\alpha$, by
\begin{equation}\label{eq:bundleMap}
     \hat{\beta}_z:T_zM \to T^*_zM,\quad X \mapsto \alpha_z = \iota_X \beta_z .
\end{equation}
From the fact that $\beta$ is maximal rank, it follows that its kernel,
\begin{equation}\label{eq:BetaKernel}
    \ker \hat{\beta}_z =: \{X \in T_zM : \iota_X \beta_z = 0\},
\end{equation}
is a one-dimensional subspace at each point $z\in M$. 

Let $\Omega$ denote a volume form on $M$. Such a form always exists since $M$ is orientable. As all volume forms are nondegenerate, there is a unique, non-vanishing $B\in \mathfrak{X}(M)$ such that
\begin{equation}\label{eq:betaDef}
   \iota_B \Omega = \beta.
\end{equation}
Noting that $\iota_B \beta = \iota_B\iota_B\Omega = 0$,
the vector field lies in $\ker \hat\beta_z$ for each $z \in M$ and is a Hamiltonian vector field for $\beta$,  \eqref{eq:Hamilton}, with trivial Hamiltonian, $H = 0$.

The following lemma provides useful properties of the magnetic field $B$.
\begin{lemma}\label{lem:BProperties}
    The magnetic vector field satisfies the following properties:
    \begin{enumerate}
        \item $B$ is divergence free, $\L_B \Omega = 0$
        \item $B$ preserves the presymplectic form $\beta$, i.e., $\L_B \beta = 0$.
    \end{enumerate}
\end{lemma}
\begin{proof}
    Using Cartan's magic formula \cite{MacKay20} and the fact that $\beta$ and $\Omega$ are closed, it is easy to see that the flow generated by $B$ is volume preserving:
    \[ \L_B \Omega = \iota_B d\Omega + d\iota_B\Omega = d \beta = 0. \]
    Similarly since $B$ is in the kernel of $\beta$, applying the Lie derivative $\L_B$ to $\beta$ yields,
    \[\L_B \beta = \iota_B d\beta + d \iota_B\beta = 0. \]
\end{proof}

Defining the vector field $B$ by \eqref{eq:betaDef} is not the only way to obtain dynamics from the presymplectic form $\beta$. Indeed, through the map $\hat{\beta}$, \eqref{eq:bundleMap}, a vector field $X$ can be generated from a one-form $\alpha$ provided we can solve $\hat{\beta}(X) = \alpha$. However, since $\ker \hat{\beta} $ is non trivial, the existence of a corresponding vector field $X$ for every choice of $\alpha$ is not guaranteed. If there does exist a solution, following \cite{Enriquez99}, we may define two related concepts:
\begin{definition}[Presymplectic and Hamiltonian Vector Fields]
    A vector field $X\in \mathfrak{X}(M)$ is \emph{presymplectic with respect to $\beta$} or is \emph{locally Hamiltonian} if there exists a closed one-form
    $\alpha$ such that $\iota_X \beta = \alpha$. In this case, we call $(X,\alpha)$ a \emph{presymplectic pair}.\\
    If, in addition, $\alpha$ is exact, that is $\alpha = -dH$ for some $H\in C^\infty(M)$, then $X$ is a Hamiltonian vector field for the Hamiltonian $H$, and $(X,H)$ is a \emph{Hamiltonian pair}.
\end{definition}

Since $\L_X \beta = \iota_X d\beta + d \iota_X \beta =  d\alpha = 0$, every presymplectic vector field, and hence every Hamiltonian vector field, preserves the presymplectic form.

Whilst every Hamiltonian pair is presymplectic, the converse is not necessarily true. It is true that, given a presymplectic pair, $B$ preserves $\alpha$.
\begin{lemma}\label{lem:presym}
    Suppose that $(X,\alpha)$ is a presymplectic pair for $\beta$. Then
    \begin{enumerate}
        \item $\iota_B \alpha = 0$;
        \item $\L_B \alpha =  0$;
        \item $\beta\wedge \alpha = 0 $.
    \end{enumerate}
    Moreover, if $(X,H)$ is a Hamiltonian pair, $B$ is tangent to surfaces of constant $H$, that is $\ \L_B H = 0$.
\end{lemma}
\begin{proof}
    By definition, $\iota_X\beta = \alpha$, and then \eqref{eq:betaDef} gives $\iota_B\alpha = \iota_B\iota_X\beta = 0$. Taking the Lie derivative  of $\alpha$ then gives
    \[\ \L_B \alpha = \iota_B d\alpha + d \iota_B\alpha = 0, \]
    using the fact that $\alpha$ is closed.
    Now $\beta \wedge \beta = 0$, so $0 = \iota_X (\beta \wedge \beta)
    =(\iota_X\beta)\wedge \beta + \beta \wedge \iota_X\beta$.
    Therefore $\beta\wedge \alpha = \beta\wedge \iota_X\beta = -(\iota_X\beta)\wedge \beta = -\alpha\wedge\beta$. 
    Since on a three-manifold two-forms commute with one-forms, this finally implies $\beta\wedge\alpha = 0$.
    
    Finally, if $\alpha = dH$ for some $H\in C^\infty(M)$ then we have $0=\iota_B \alpha = \iota_B dH = \L_B H$, concluding the proof.
\end{proof}

It is natural to ask whether the converse of  \cref{lem:presym}(1) is true; that is, provided $\iota_B\alpha = 0$ 
and $d\alpha=0$, are we guaranteed that there exists a presymplectic vector field $X$
that makes $(X,\alpha)$ a presymplectic pair? The following proposition provides an affirmative answer.
\begin{lemma}\label{lem:PresympCriteria}
    Suppose $M$ is an orientable three-manifold. Then $(X,\alpha)$ is presymplectic with respect to $\beta$ if and only if $ \iota_B\alpha = 0 $.
\end{lemma}

\begin{proof}
    Since the only if part is \cref{lem:presym}(1), we only need to show that, given a nonzero, closed one-form $\alpha$ with $\iota_B\alpha = 0$, there exists a vector field $X$ such that $\iota_X \beta = \alpha$. Without loss of generality, assume $M$ is equipped with a Riemannian metric $g$. Then, using the one-form $B^\flat$, recall \eqref{eq:flatDef}, we define a map from vector fields to one-forms,
    \begin{equation}\label{eq:hatb}
        \hat{b}: X \mapsto \iota_X \beta + g(X,B) B^\flat,
    \end{equation}
    This map is a bundle isomorphism: indeed, injectivity is seen by noting $\hat{b}(X) = 0$ only if $g(X,B) = 0$ and $\iota_X \beta = 0$, which is only possible if $X = 0$. Now, for surjectivity, we must have a solution to $\iota_X\beta + g(X,B)B^\flat = \alpha$ for any $\alpha$. Contracting with $B$ yields
    \[ g(X,B)|B|^2 = \iota_B \alpha \implies g(X,B) = |B|^{-2} \iota_B \alpha. \]
    Taking the wedge product with $B^\flat$ yields,
    \[ (\iota_X\beta) \wedge B^\flat = \alpha \wedge B^\flat \implies \iota_X( \beta\wedge B^\flat) = \alpha\wedge B^\flat + (|B|^{-2}\iota_B\alpha) \beta.\]
    As $\beta\wedge B^\flat$ is a volume form, then this equation must have a solution $X$. Thus $\hat{b}$ is an isomorphism as claimed.
    
    Suppose that $\iota_B \alpha = 0$. Then $X = \hat{b}^{-1}(\alpha)$  is the unique vector field satisfying $g(X,B) = 0$ and $\iota_X \beta = \alpha$.
\end{proof}



There are some interesting differences between Hamiltonian vector fields on a presymplectic manifold $(M,\beta)$ and Hamiltonian vector fields on a symplectic manifold $(N,\omega)$ (see \cref{sec:Weak}). Recall that a symplectic two-form $\omega$ is nondegenerate, and therefore given any function $H\in C^\infty(M)$ there is unique, Hamiltonian vector field $X_H$ by \eqref{eq:Hamilton}.
In the presymplectic case, such a vector field need not exist. Even if it does exist,
it is not unique. Indeed, any $X + f B$ for $f \in C^\infty(M)$ is also 
a Hamiltonian vector field. 

Perhaps more surprisingly, unlike the symplectic case, not all presymplectic Hamiltonian vector fields are volume preserving!
\begin{lemma}\label{lem:divIsCommute}
    A Hamiltonian vector field $X$ is volume preserving if and only if $[X,B] = 0$
\end{lemma}
\begin{proof}
    First, note that since $\Omega$ is a three-form on a three-manifold, the divergence is defined
    by $\ \L_X\Omega = (\Div X) \Omega$  so that $ \ \L_X \Omega = 0 \iff \Div X = 0 $.
    Now, since $\Omega$ is nondegenerate, $\iota_{[B,X]} \Omega = 0 \iff [B,X] = 0 $. Moreover,
    \begin{align*}
        \iota_{[B,X]} \Omega &= \iota_B \ \L_X \Omega - \L_X \iota_B \Omega \\
        &= (\Div X )\iota_B \Omega + \L_X \beta \\
        &= (\Div X )\beta .
    \end{align*}
    Thus, $\L_X\Omega = 0 \iff [X,B] = 0$.
\end{proof}

\subsection{Integrable presymplectic systems} \label{sec:Integrability}

As we argued in \S\ref{sec:Intro}, 
the geometric formulation of integrability requires both an invariant and a corresponding symmetry.

\begin{definition}[Integrability]\label{def:integrable}
    Let $M$ be an orientable three-manifold with presymplectic form $\beta$.
    We say that the presymplectic form $\beta$ is integrable if there exists a
    volume form $\Omega$ and a Hamiltonian pair $(J,p)$ so that $B$ is the unique vector field that satisfies $\iota_B\Omega = \beta$, and
    \[
        [B,J] = 0.
    \]
    We refer to the $(B,J,p, \Omega)$ (or equivalently $(\beta,J,p,\Omega)$) as an \textit{integrable presymplectic system}.     
\end{definition}
From \cref{lem:presym,lem:divIsCommute}, integrability of $B$ is equivalent to the existence of a Hamiltonian pair $(J,p)$ such that the flow of $J$ is volume preserving.

For the magnetohydrostatics case, we require, in addition, that $J$ is the current:
\begin{definition}[MHS integrable]\label{def:MHSIntegrable}
Let $(M,g)$ be a Riemannian three-manifold with volume form $\Omega$ induced by the metric tensor $g$. An integrable presymplectic system $(B,J,p,\Omega)$ on $M$ is  \emph{MHS integrable} if
\[
    \iota_J\Omega = d\iota_B g \equiv dB^\flat.
\]
\end{definition}

The goal of this section is to understand the possible topologies of the level sets of $p$ 
for an integrable presymplectic system on a compact three-manifold $M$. In  \S\ref{sec:NormalForm}, 
we will use this to obtain normal form coordinates near points of $M$ where $dp$ vanishes.

To avoid pathologies, we will assume that $p$ is \emph{proper}, that is, 
that inverse images of compact sets are compact. It is useful 
to consider $p$ as generating a foliation of $M$, the leaves of which are the level 
sets $p^{-1}(c)$, $c\in p(M)$. A leaf of $p^{-1}(c)$ is said to be \emph{regular} if $dp\neq 0$ 
for all points on $p^{-1}(c)$ and \textit{singular} if $dp=0$ for some point on $p^{-1}(c)$. Due to the fact that $p$ is proper, both the regular and singular leaves are compact. 

It is well known that the connected components of regular leaves of $p$ are 
diffeomorphic to $\T^2$ \cite[Sec. 49]{Arnold78}.
The proof of this fact goes roughly as follows. Let $\Phi^B_{t}$, and $\Phi^J_{t}$ be the flows of 
$B$, and $J$, respectively. Since $[B,J] = 0$, the two-parameter group generated
by $\Phi_{t_1,t_2} := \Phi^B_{t_1}\circ\Phi^J_{t_2}$, is Abelian, and so can be thought of as 
an $\R^2$-action on the manifold $M$. Since both $B$ and $J$ preserve $p$, it follows 
that the orbits of this action are contained in the leaves of $p$. For each regular leaf, the $\R^2$-action is 
non-singular and thus the orbit of a point on the leaf must be either 
diffeomorphic to $\R^2, \T\times \R$ or $\T^2$. Such an orbit must be closed; indeed, if there were some point in the closure that is not in the orbit, then the $\R^2$ action would be singular at this point, contradicting the regularity of the leaf. Hence the orbit must be diffeomorphic to $\T^2$ and the result follows.

By contrast, the singular leaves, which have non-empty intersections with
\[
    \crit(p) = \{ z \in M : dp_z = 0\}, 
\]
i.e., the set of critical points of $p$, can be arbitrarily complex. A sketch showing some possible regular and singular leaves is given in \cref{fig:comparison}. We will focus our attention on transversely nondegenerate critical points. 

\begin{definition}
    The set of \emph{nondegenerate} critical points of $p$, $\crit_0(p)$, contains all 
    $z\in\crit(p)$ such that, on any section $\Sigma \ni z$ transverse to $B$, the restriction $p|_\Sigma$ of $p$ to $\Sigma$ is a Morse function in a neighborhood of $z$.
\end{definition}

\begin{figure}[ht]
    \centering
    \includegraphics[width=0.5\linewidth]{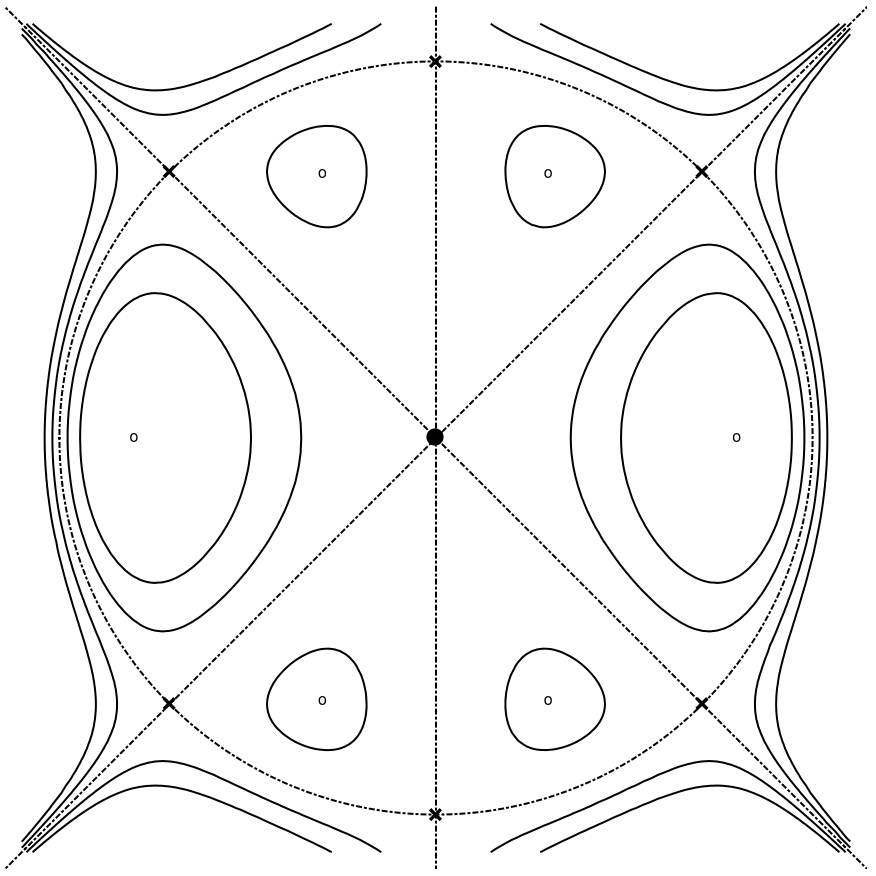}
    \caption{Example of an intersection between trajectories of $B$ and a transverse section $\Sigma$. The `x' points represent hyperbolic singular orbits while the  `o' points represent elliptic. The central, $\bullet$ point corresponds to a degenerate orbit. The dashed curves are  singular leaves containing several hyperbolic points. The tori encircling tori encircling the elliptic orbits are regular leaves. }
    \label{fig:comparison}
\end{figure}

The topological structure of $\crit_0(p)$ for two degree of freedom Hamiltonian systems was given in \cite{deVediere00}. The proof is easily adapted to the presymplectic case.

\begin{lemma}
    If $p$ is a proper map, then the set of nondegenerate critical points, $\crit_0(p)$, is the union of finitely many, disjoint, periodic orbits of $B$. 
\end{lemma}
\begin{proof}
    Since $B$ is non-vanishing,  the flow-box theorem guarantees that there is a neighborhood $V$ of any $z\in\crit_0(p)$ so that $B$ has section $\Sigma$ containing $z$. Moreover, $V$ is diffeomorphic to $\Sigma\times I$ where $I$ is some small interval containing $0$, and each trajectory of $B$ is given by $s\times I$ for some $s\in\Sigma$. As $p$ is constant under the flow of $B$, then it follows $p(s\times I) = p(s,0)$ for each $s\in \Sigma$.  Furthermore, the rank of the quadratic form on $z\times I$ induced by $p$ is constant, and $dp(s,t) = 0$ if and only if $s = z$. Hence, $z\times I$ is diffeomorphic to $V\cap \crit_0(p)$, and $\crit_0(p)$ is a smooth submanifold of $M$.
    
    As $z$ is nondegenerate, a simple corollary of the Morse-Bott lemma (see \cref{thm:Morse-Bott}) guarantees that $z$ is an isolated critical point on $\Sigma$. Together with the compactness of $M$, this guarantees that $\crit_0(p)$ is compact. Hence, the orbit of each point $z\in\crit_0(p)$ is a compact one-dimensional submanifold of $\crit_0(p)$ and thus a periodic orbit. It can then be concluded that $\crit_0(p)$ is the union of disjoint periodic orbits $\gamma_i $, i.e. $\crit_0(p) = \cup_i \gamma_i$. Finally, as $p$ is proper, there can only be finitely many such $\gamma_i$.
\end{proof}

The fact that the set of nondegenerate critical points $\crit_0(p)$ is constituted by periodic orbits warrants the following definition.

\begin{definition}[Magnetic Axes] \label{def:SingularOrbit}
    Orbits $\gamma_i$ in $\crit_0(p)$ are called \emph{singular orbits} of the $\mathbb{R}^2$ action generated by $B$ and $J$. In the context of magnetic fields, the singular orbits are referred to as \emph{nondegenerate magnetic axes}.
\end{definition} 
Due to the nondegeneracy condition, singular orbits come in two flavors. If $p|_\Sigma$ is locally a maximum or minimum at $z_i\in \gamma_i\cap\Sigma$, then the singular orbit $\gamma_i$ is an elliptic periodic orbit. In this case, the singular leaf $\Gamma = p^{-1}(p(z_i))$ has a connected component that is merely the single periodic orbit $\gamma_i$. In the alternative case, when $p|_\Sigma$ is not a local extremum at $z_i$, the orbit $\gamma_i$ is a hyperbolic orbit.

For an integrable presymplectic system, the singular orbits can be classified using
the normal Hessian of $p$:

\begin{definition}[Elliptic/Hyperbolic singular orbits]\label{def:ellipticAxis}
An \emph{elliptic/hyperbolic} singular orbit for an integrable presymplectic system $(B,J,p,\Omega)$ is a smooth closed curve $\gamma$ such that the normal Hessian of $D^2_\perp p$ along $\gamma$ is sign-definite/indefinite.
\end{definition}
In the hyperbolic case a connected component of a leaf may contain several singular orbits, which are perhaps degenerate, recall \cref{fig:comparison}. Provided the connected component does not contain any degenerate orbits, the following generalization of a result for two degree-of-freedom Hamiltonian systems, \cite{deVediere00}, reveals how the hyperbolic orbits may be connected.

\begin{lemma}
    Let $\Gamma$ be a connected component of a singular leaf containing only hyperbolic orbits $\gamma_i$. Then $\Gamma\setminus \cup_i\gamma_i$ is the union of 
    hetero- or homoclinic orbits
    $\gamma_{i,j}^k$.
\end{lemma}
\begin{proof}
    As in the topological argument for the regular leaves of $p$, we have an $\R^2$-action generated from the flows of $B$ and $J$. Any point $u\in\Gamma\setminus \cup_i\gamma_i$ is a regular point of this action, thus, its orbit under the $\R^2$ action must be diffeomorphic to either $\R^2, \T\times\R$ or $\T^2$. However, by compactness of $M$, and since the orbit of $u$ must contain at least one periodic orbit $\gamma_i$ in its closure, it cannot be diffeomorphic to $\R^2$ or $\T^2$. Thus, $\Gamma\setminus\cup_i\gamma_i$ is the union of orbits diffeomorphic to $\T\times \R$. 
\end{proof}

As an immediate consequence of the lemma, the closure of any $\gamma_{i,j}^k$ must contain either one or two singular orbits $\gamma_i,\gamma_j$. Denote by $\gamma_{i,j} = \cup_k \gamma^k_{i,j}$ the set of connecting orbits between $\gamma_i$ and $\gamma_j$. Each $\gamma_i$ has an associated stable $W^s(\gamma_i)$ and unstable $W^u(\gamma_i)$ manifold. The local stable and unstable manifolds are each 1D vector bundles over $\gamma_i \cong \T$. Up to diffeomorphism, there are only two possibilities for such a vector bundle: it is either a trivial or a Mobi\"{u}s bundle. Moreover both the stable and unstable manifolds must have the same type. This motivates the following definition.

\begin{definition}
    If $W^s(\gamma_i)$ is a trivial vector bundle over $\gamma_i$, then we say $\gamma_i$ is \emph{direct hyperbolic}. If $W^s(\gamma_i)$ is a Mobi\"{u}s bundle then we say $\gamma_i$ is \emph{reflection hyperbolic}. 
\end{definition}

Sketches of these cases are shown in \cref{fig:HyperbolicCases}.

\begin{figure}[ht]
    \centering
    \begin{subfigure}{0.48\linewidth}
        \includegraphics[width=\linewidth]{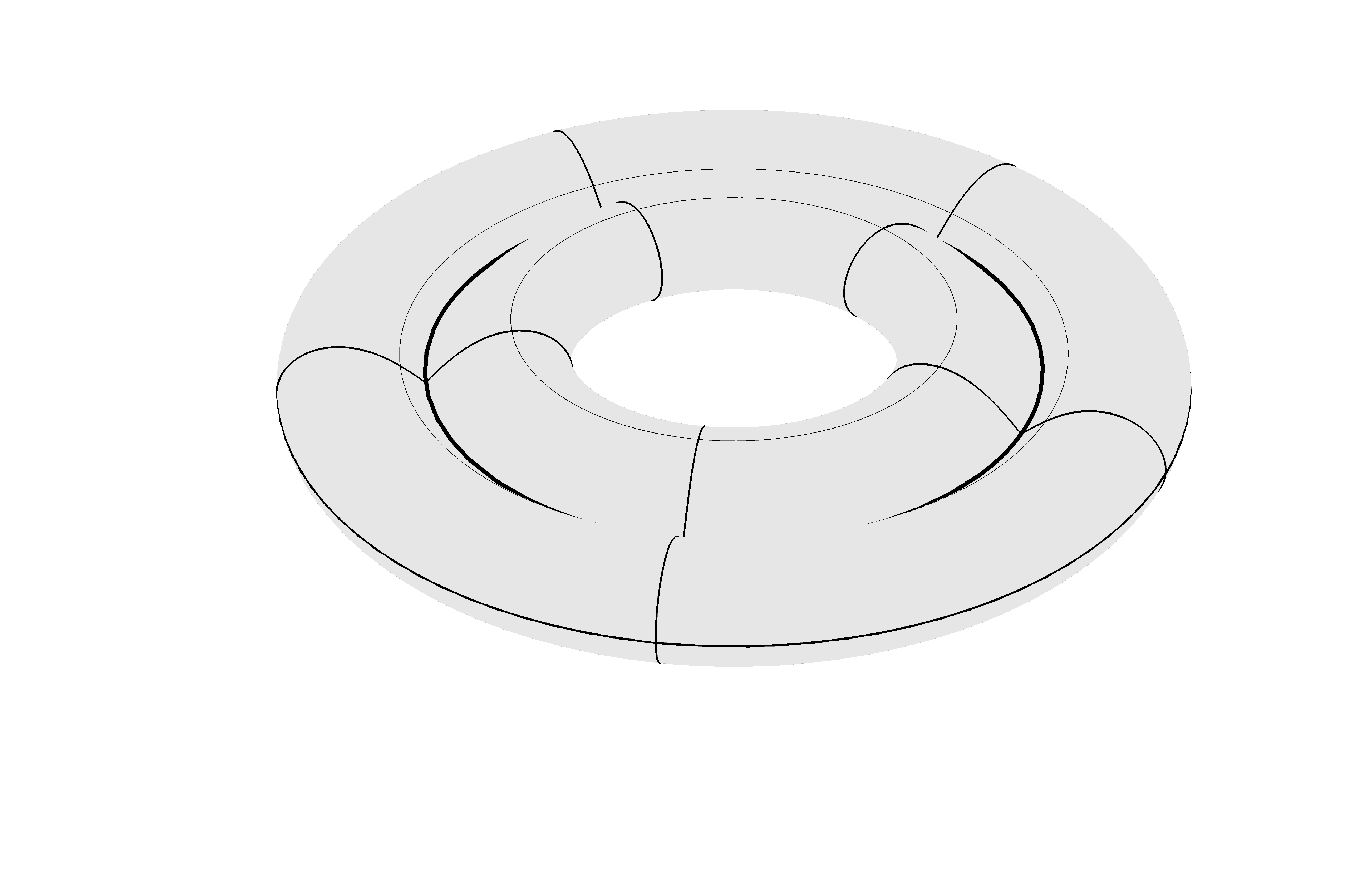}
        \subcaption{Direct Hyperbolic with 0 twist}
    \end{subfigure}
    \begin{subfigure}{0.48\linewidth}
        \includegraphics[width=\textwidth]{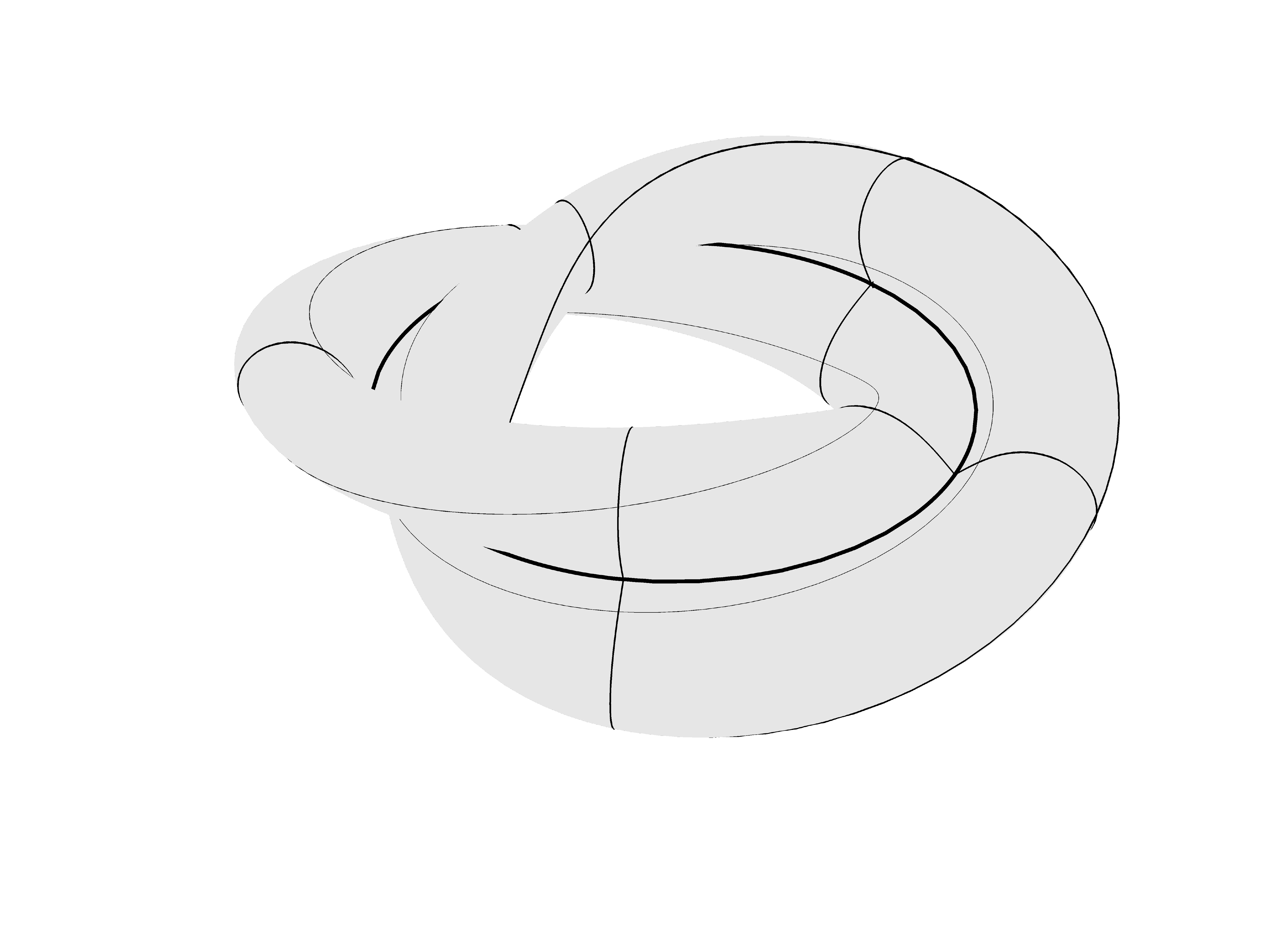}
        \subcaption{Direct Hyperbolic with 1 twist}
    \end{subfigure}
    \begin{subfigure}{0.48\linewidth}
        \includegraphics[width=\textwidth]{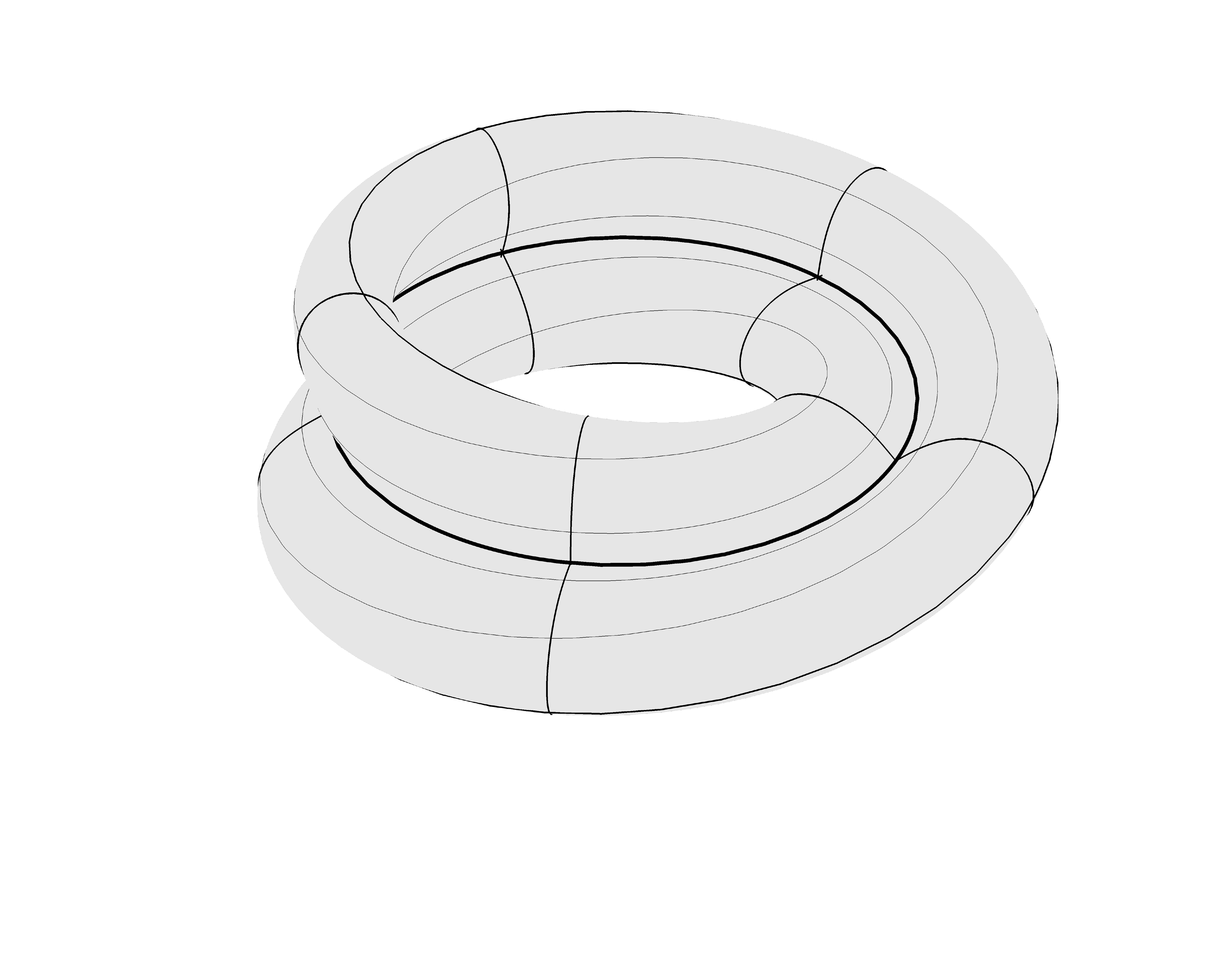}
        \subcaption{Reflection Hyperbolic with 1/2 twist}
    \end{subfigure}
    \begin{subfigure}{0.48\linewidth}
        \includegraphics[width=\textwidth]{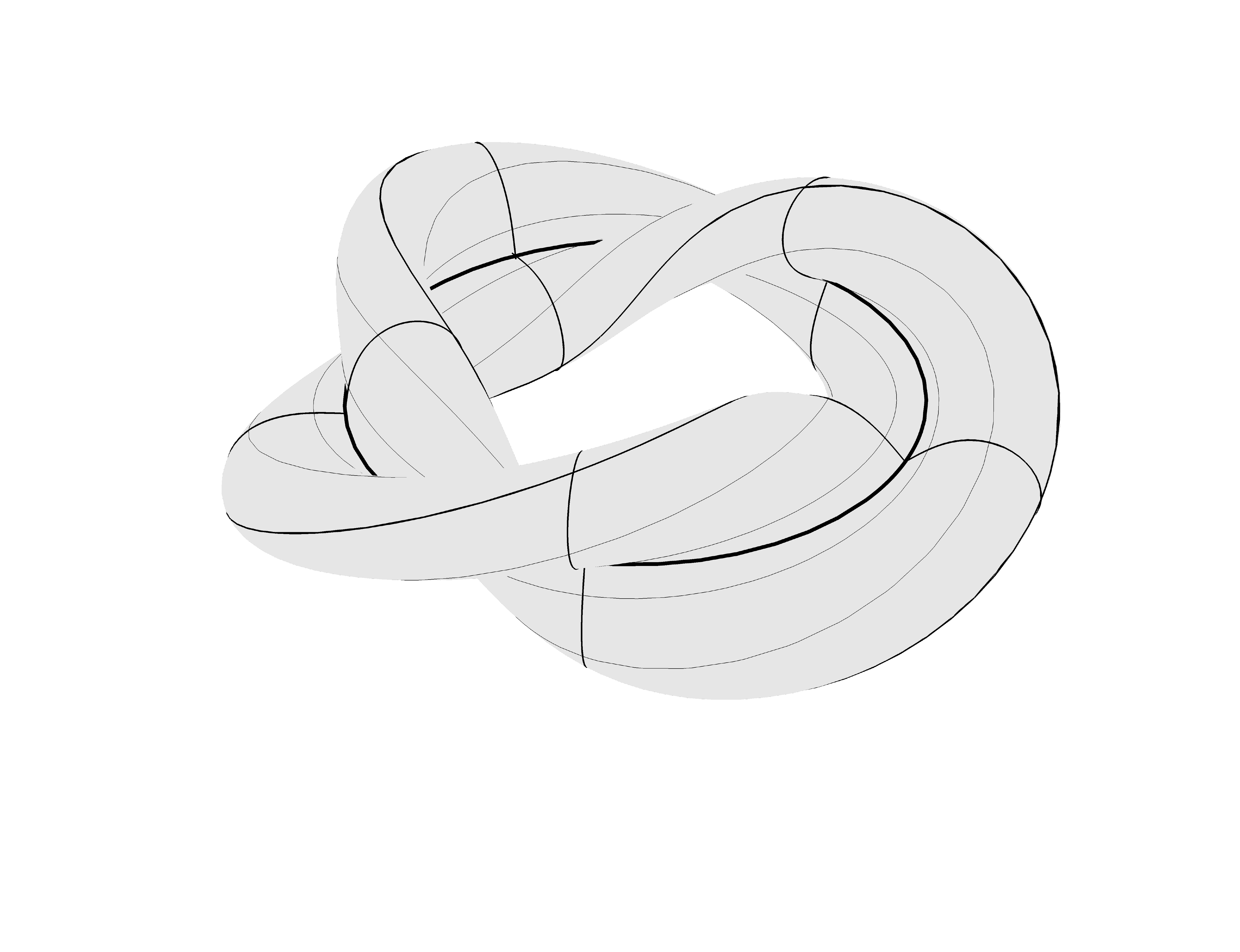}
        \subcaption{Reflection Hyperbolic with 3/2 twist}
    \end{subfigure}
    \caption{Examples of direct and reflection hyperbolic orbits.}
    \label{fig:HyperbolicCases}
\end{figure}

Since the regular leaves of an integrable system are diffeomorphic to $\T^2$, a skeleton of the phase space
can be built from understanding of the singular leaves and how they intertwine. This is particularly useful in understanding toroidal confinement devices with ``divertors". For example, a device with a central elliptic magnetic axis and a hyperbolic axis on the outer edge to divert plasma away from the center, as sketched in \cref{fig:Divertor}.

\begin{figure}
    \centering
    \includegraphics[width=0.7\textwidth]{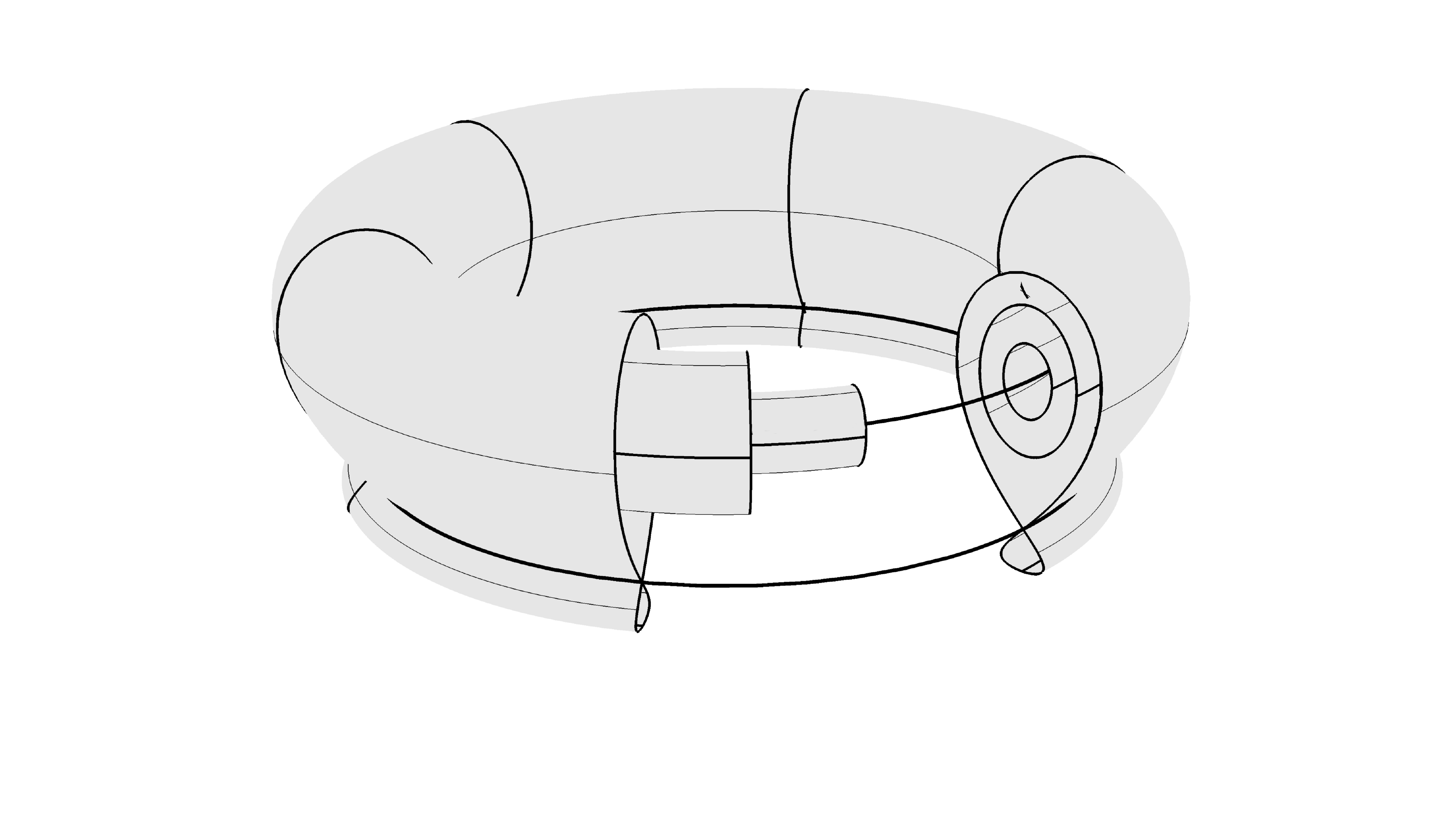}
    \caption{Elliptic and direct hyperbolic singular orbits composed to give a ``divertor".}
    \label{fig:Divertor}
\end{figure}

\section{Normal Form Coordinates for Integrable Presymplectic Systems}\label{sec:NormalForm}

In a study of a given geometric structure, it is often paramount to form the ``simplest'' possible coordinates. Of course, ``simplest'' is often a matter of taste; perhaps the coordinates accentuate some geometric properties or the lack thereof, or reveal some algebraic structure. For an integrable symplectic system, such coordinates are known as \textit{action-angle} coordinates and these exist even in a neighborhood of singular leaves.  As we explain below, certain results for the symplectic case can be modified to show that similar---action-angle---coordinates exist near a regular leaf of an integrable presymplectic system, giving \cref{thm:AML}.

We then obtain normal form coordinates near a nondegenerate magnetic axis in \cref{thm:SingularCoordinates}. This theorem will be proved using different techniques for the elliptic and hyperbolic  cases in \S\ref{sec:ellipticproof}, and \S\ref{sec:hyperbolicproof}, respectively.
As we will see a simple, direct argument based on what is called \textit{Moser's trick} is available in the elliptic case. By contrast, the hyperbolic case requires invoking a deeper result from symplectic geometry.

A first result establishes the existence of a symmetry group near a regular surface when $p$ is a proper map.

\begin{thm}[\cite{zungConceptualApproachProblem2018}]\label{thm:Zung}
    Suppose that $M$ is a compact three-manifold and that $(\beta,J, p,\Omega)$ is an integrable presymplectic system such that $p$ is proper. Then, in a neighborhood $\U(\Lambda)$ of a regular leaf $\Lambda$ there exists, up to automorphisms of $\T^2$, a unique, free $\T^2$-action that preserves $\beta$ and whose orbits are regular level sets of $p$ in $\U$. 
\end{thm}

As pointed out in \cite{zungConceptualApproachProblem2018}, the existence of this group is fundamental to the existence of the action-angle variables.  Indeed, given this group, one can find coordinates with straight field lines.

\begin{thm}[Arnol'd-Liouville-Mineur \cite{Zung05}]\label{thm:AML}
    Suppose $M$ is a compact three-manifold and that $(\beta, J, p,\Omega)$ is an integrable presymplectic system with a proper Hamiltonian $p$. Then, for each connected component, $\Lambda_i$, of a regular leaf $\Lambda$ of $p$:
    \begin{enumerate}
        \item $\Lambda_i$ is diffeomorphic to $\T^2$.
        \item There is a neighborhood $\U(\Lambda_i)$ and smooth coordinates $(\psi,\theta_1,\theta_2)\in\R\times\T^2$ such that,
        \begin{equation}\label{eq:AMLForm}
        \begin{split}
            p &= p(\psi), \\
            \beta &= d\Psi_1(\psi) \wedge d\theta_1 - d\Psi_2(\psi) \wedge d\theta_2 .
        \end{split}
        \end{equation}
    \end{enumerate}
\end{thm}

These coordinates reveal the inherent geometry and algebra of the integrable system in the neighborhood of a regular leaf. Writing $B = B^{\theta_1}\partial_{\theta_1} + B^{\theta_2} \partial_{\theta_2} + B^\psi\partial_\psi$, then $\iota_B \beta = 0$ immediately implies that
\[
    B^\psi = 0,
    \qquad 
    \frac{B^{\theta_1}}{B^{\theta_2}} = \frac{\Psi_2^\prime}{\Psi_1^\prime} = \iotab(\psi).
\]
Thus, $\psi$ is constant under $B$ and the field lines are straight with pitch $\iotab(\psi)$ on level sets of $\psi$. Moreover, \[[\partial_{\theta_1},B] = [\partial_{\theta_2},B] = [\partial_{\theta_1},\partial_{\theta_2}] = 0. \]
Consequently, the vector fields $\partial_{\theta_1},\partial_{\theta_2}$ generate the $\T^2$-action of \cref{thm:Zung} on the neighborhood $\U(\Lambda)$ that preserves $B$ and hence $\beta$.

The functions $\Psi_1,\Psi_2$ in \eqref{eq:AMLForm} are known as actions. In fact, it can be shown that each $(\Psi_i,\partial_{\theta_i})$ is a Hamiltonian pair for $\beta$. Mineur \cite{mineur1935systemes} gave a way to compute the actions. Let $ c_1,c_2$ be any pair of generators of the fundamental group of the chosen regular leaf and suppose we have $\beta = d\alpha$. Then,
\[ 
    \Psi_1 = \frac{1}{2\pi}\int_{c_1} \alpha,\qquad 
    \Psi_2 = -\frac{1}{2\pi}\int_{c_2} \alpha. 
\]
The actions $\Psi_1$ and $\Psi_2$ and their respective angles $\theta_1,\theta_2$ are dependent on the homology of $c_1,c_2$.
Due to this dependence, it is immediate that $\Psi_1,\Psi_2$ are not unique. In fact, this non-uniqueness will be leveraged in \S\ref{sec:Applications}.

The primary goal of this section is to prove an analog of \cref{thm:AML} near nondegenerate singular orbits. This result, previously stated as \cref{sec2_thm1}, becomes more formally:

\begin{thm}\label{thm:SingularCoordinates}
    Suppose that $M$ is a compact three-manifold and assume that $(\beta,J,p,\Omega)$ is a smooth (resp. analytic) integrable presymplectic system with a proper Hamiltonian $p$. Then, there exists a tubular neighborhood $\U$ of any nondegenerate magnetic axis $\gamma$, 
    a smooth (resp. analytic) transformation $\Phi:D^2\times\R \to \U $ and coordinates $(x,y,\phi) \in D^2\times\R$ such that,
    \begin{equation}\label{eq:SingularCoord}
    \begin{split}
            \Phi^*p &= P(\psi) \\
            \Phi^*\beta &= dy\wedge dx - d\Psi_P(\psi)\wedge d\phi,
    \end{split}
    \end{equation}
    where $\psi = \tfrac12 (x^2+ \epsilon y^2)$ and $P,\Psi_P$ are smooth (resp. analytic). 
    
    Here, $\epsilon = -1$ corresponds to the hyperbolic and $\epsilon = +1$ to the elliptic case.
    Moreover, $\Phi$ is a diffeomorphism if $\gamma$ is elliptic or direct hyperbolic, and a double cover if it is reflection hyperbolic.
\end{thm}

\begin{remark}
    The coordinates $(x,y,\phi)$ of \cref{thm:SingularCoordinates} are \emph{normal form coordinates}. If we think of $p$ as a nonautonomous Hamiltonian in a neighborhood of $\gamma$, it becomes clear that these coordinates put $p$ in Birkhoff normal form \cite{Arnold78}. While the Birkhoff normal form is in general only a formal construction, \cref{thm:SingularCoordinates} guarantees that this can be made smooth, or analytic.
\end{remark}

\begin{remark}
The coordinates $(x,y,\phi)$ above are certainly not unique. For suppose we have found one such set of coordinates and let $p^*$ and $\beta^*$ to be the forms in \eqref{eq:SingularCoord}. If $u$ is a vector field such that $\L_up^* = 0$ and $\L_u\beta^* = 0$ then $(\overline{x},\overline{y},\overline{\phi}) = \exp(u)(x,y,\phi)$ is another set of normalizing coordinates. There are many such $u$, including any vector fields of the form $u = f\,B$, where $f$ is some scalar function. This freedom is not atypical in normal form theory. We will use the freedom in \S\ref{sec:Applications} to construct near-axis Hamada and Boozer coordinates.
\end{remark}

Separate proofs will be given for the elliptic and hyperbolic cases in \S\ref{sec:ellipticproof} and \S\ref{sec:hyperbolicproof}, respectively, due to their differing topologies.
The proof strategy for the elliptic case is reminiscent of methods used
in the analogous theorem for integrable symplectic systems by Eliasson \cite{Eliasson90} in that it makes extensive use of the \textit{Moser trick}. The key idea for the hyperbolic case is to embed the integrable presymplectic system into an integrable symplectic system in 4D. Then, one can lean on results in the symplectic setting to avoid much of the heavy lifting.

Both proofs rely on a powerful theorem from Morse theory. First, we recall the definition of a nondegenerate critical manifold. Wherever necessary, we assume that $M$ is equipped with a Riemannian metric $g$.

\begin{definition}
      A submanifold $\C$ is called a \emph{Morse-Bott critical manifold} of a function $f\in C^\infty(M)$ if $\C \subset \crit_0(f)$, so that the restriction $ D_\perp^2f =  D^2f|_{N\C}$, of the Hessian to the normal bundle
    \[
        N\C = T_z M/T_z \C
    \]
    is nondegenerate. 
\end{definition}
\begin{thm}[Morse-Bott]\label{thm:Morse-Bott}
    If $ \C\subset M$ is a Morse-Bott critical manifold of  $f \in C^\infty(M)$, then there is a neighborhood $\U$ of the zero section of the normal bundle $ N\C $ and an open embedding $ \Phi:\U\to M$ such that $ \Phi|_\C = Id $ and 
    \[
        \Phi^*f(c,n) = f(\C) + \langle n,D^2_\perp f n \rangle,
    \]
    where 
    $(c,n)\in\U$, and $\langle,\rangle$ is an inner product on $N\C$.
\end{thm}

With this theorem we can immediately prove the following corollary for the case that the critical
manifold is a nondegenerate singular orbit of an integrable presymplectic system.
\begin{lemma}\label{lem:ThmA}
    In a cover $ U = D^2 \times \T$ of the tubular neighborhood $\U(\gamma)$ of \cref{thm:SingularCoordinates} there exist coordinates $(x,y,\phi)$ such that
    \begin{enumerate}
        \item $p(x,y,\phi) = p(\gamma) +\tfrac12 \,c\,(x^2 + \epsilon \,y^2)$
        \item $\beta(0,0,\phi) = dy\wedge dx$,
    \end{enumerate}
    where $c \neq 0$ and $\epsilon = 1$ ($-1$) if $\gamma$ is elliptic (hyperbolic). The cover is injective if $\gamma$ is elliptic or direct hyperbolic and is a double cover if it is reflection hyperbolic.
\end{lemma}
\begin{proof}
  From \cref{thm:Morse-Bott} we are guaranteed a map $\Phi_1$ from a neighborhood of the zero section of the normal bundle $N\gamma$ to $\U(\gamma)$ such that $\Phi_1^* p$  has this form
  near each point on $\gamma$. In the elliptic and direct hyperbolic case, the normal bundle $N\gamma$ is trivializable and hence $N\gamma = D^2 \times \T$. By rotating and linearly scaling the normal space if necessary with a transformation $\Phi_2$, and using $\Phi = \Phi_2\circ\Phi_1$, we obtain the Hessian  $\tfrac12\,C\,(X^2+\epsilon Y^2)$, with $C$ nonzero. 

  In the reflection hyperbolic case, the normal bundle $N\gamma$ is not trivializable because the (un)stable manifolds are nonorientable. Each of the manifolds are homeomorphic to a M\"{o}bius band, and consequently, by taking a double cover, we are able to orient the spaces. The remainder of the argument follows as in the elliptic and direct hyperbolic case.

  We will finally construct coordinates $(x,y,\phi)$ such that $\beta(0,0,\phi) =  dy\wedge dx$. First, note that $\ker\hat\beta|_\gamma = T\gamma$ and thus $\beta(0,0,\phi) = f(\phi) dY\wedge dX$ where $f(\phi)$ is nowhere-vanishing. Let us examine the transverse linearization $\tilde{\beta}$ of $\beta$ about $\gamma$. This linearization is of the form,
  \[
        \tilde{\beta} = f(\phi) dY\wedge dX + (a_1(\phi) X+ 
        a_2(\phi) Y) dX \wedge d\phi + (a_3(\phi)X + a_4(\phi) Y) dY\wedge d\phi 
    \]
    From the condition $\beta\wedge dp = 0 \implies \tilde{\beta}\wedge dp = 0$ we have that $a_1 = a_4$ and $a_2 = a_3 = 0$. Moreover, since $\tilde{\beta}$ is presymplectic, then $d\tilde{\beta} = 0$. This implies $f^\prime +a_2 - a_3 =0$, thus $f^\prime = 0$. The result now follows with $c = C/|f|$ by appropriately scaling and permuting $X$ and $Y$.
\end{proof}

\subsection{The elliptic case: Moser's trick}\label{sec:ellipticproof}

In this section, we use Moser's trick to prove \cref{thm:SingularCoordinates} for the smooth elliptic case.

\begin{proof}[Proof of \cref{thm:SingularCoordinates}]
By \cref{lem:ThmA}, there is a tubular neighborhood $\U$ containing $\gamma$ and coordinates $(X,Y,Z)\in D^2\times \mathbb{T}$ such that $p =p(\gamma) + \tfrac12\,c\,(X^2 + Y^2) $ and $\beta = \beta_{YX}\,dY\wedge dX + \beta_{YZ}\,dY\wedge dZ + \beta_{XZ}\,dX\wedge dZ$ such that $\beta_{YX}(0,0,Z) = 1$.
Define a $\mathbb{T}^2$-action $T_{\theta,\zeta}$ in this tubular neighborhood according to $T_{\theta,\zeta}(X,Y,Z) = (\cos\theta \,X + \sin\theta\,Y, \cos\theta\,Y -\sin\theta\,X,Z +\zeta)$. Let 
$\partial_\theta = \partial_{\theta}T_{\theta,0}\mid_{\theta=0}$ and $\partial_\zeta = \partial_{\zeta}T_{0,\zeta}\mid_{\zeta = 0}$ be the corresponding infinitesimal generators. 

Because $\U \cong D^2\times \T$ has trivial second de Rham cohomology, there must be a one-form $\alpha$ such that $\beta  =d\alpha$. We may therefore define the toroidal flux for a surface of fixed $p$ as
\begin{equation}\label{eq:actionIntegral}
\begin{split}
    \Psi_T(p) = \frac{1}{2\pi}\oint_{c_P}\alpha, \quad
\end{split}
\end{equation}
where the parameterized curve $c_P$ is given by 
\[
c_P(t) = \sqrt{\frac{2(p-p(\gamma))}{c}}(\cos t, -\sin t, 0).
\]
The function $\psi = \Psi_T\circ p$, 
in particular, satisfies $\psi(0,0,Z) = 0$, $d\psi(0,0,Z)  =0$, and $D_\perp^2\psi(0,0,Z) = \text{id}$. This means $\psi$ is Morse-Bott with critical manifold $\gamma$.

In light of these properties satisfied by $\psi$, the formulas
\begin{align*}
    \widetilde{x} &= X\,\sqrt{c\psi(X,Y,Z)/(p(X,Y,Z) - p(\gamma))}\\
    \widetilde{y} & = Y\,\sqrt{c\psi(X,Y,Z)/(p(X,Y,Z) - p(\gamma)))},\\
    \widetilde{\phi}& = Z,
\end{align*}
define smooth coordinates in a neighborhood of $\gamma$.
In these coordinates $\psi= \tfrac12(\widetilde{x}^2 + \widetilde{y}^2)$, by construction. Moreover using Stokes' theorem and Taylor's theorem with remainder on \eqref{eq:actionIntegral} gives
\begin{align*}
    \psi(\widetilde{x},\widetilde{y},\widetilde{\phi})
    & = \frac{1}{2\pi}\int_0^{\sqrt{\widetilde{x}^2+\widetilde{y}^2}}\int_0^{2\pi}\beta_{\widetilde{y}\widetilde{x}}(R\cos\theta,R\sin\theta,0)\,R\,dR\,d\theta\\
    & = \tfrac12(\widetilde{x}^2 + \widetilde{y}^2)\beta_{\widetilde{y}\widetilde{x}}(0,0,0) + O([\widetilde{x}^2+\widetilde{y}^2]^{3/2}).
\end{align*}
Equality of Taylor series therefore implies $\beta(0,0,\widetilde{\phi}) = d\widetilde{y}\wedge d\widetilde{x}$.

Define the poloidal flux
\begin{align}
    \Psi_P(\psi) = -\frac{1}{2\pi}\oint_{c_T}\alpha,\label{poloidal_flux_def}
\end{align}
where the parameterized curve $c_T$ is given by $c_T(t) = (\sqrt{2\psi},0,t)$. Introduce the one-form
\begin{align*}
    \alpha_* = \tfrac12(\widetilde{y}\,d\widetilde{x} - \widetilde{x}\,d\widetilde{y}) - \Psi_P(\psi)\,d\widetilde{\phi},
\end{align*}
and its differential
\begin{align*}
\beta_* =d\alpha_*= d\widetilde{y}\wedge d\widetilde{x} - \Psi^\prime_P(\psi)\,d\psi\wedge d\widetilde{\phi}.
\end{align*}
The two-form $\beta_*$ is the desired normal form for $\beta$. However, in the coordinates $(\widetilde{x},\widetilde{y},\widetilde{\phi})$, $\beta$ and $\beta_*$ are only guaranteed to agree on the singular orbit $\gamma$. We therefore seek new coordinates defined by a diffeomorphism $\Phi:D^2\times \T\rightarrow D^2\times \T:(\widetilde{x},\widetilde{y},\widetilde{\phi})\mapsto(x,y,\phi)$ in which $\beta = \beta_*$. 

We will show that $\Phi$ can be constructed as the (inverse of the) time-one flow map $\varphi_{\lambda}$ of a time-dependent vector field $\xi_\lambda$ on $D^2\times \T$. The goal is to realize the interpolation
\[
    \beta_\lambda = [1-\lambda]\beta +\lambda\,\beta_*
\]
by Lie dragging along $\varphi_\lambda$ while preserving $\psi$; that is, we require $\varphi_{\lambda}$ satisfies
\begin{equation}\label{eq:gproof}
   \varphi^{*}_{\lambda}\beta_\lambda = \beta, \qquad
   \varphi^{*}_{\lambda}\psi =\psi,
\end{equation}
for all $\lambda\in[0,1]$. If this can be done then, for the diffeomorphism $\Phi$, we could choose $\Phi=(\varphi_{1})^{-1}$
and the desired normal form would be achieved. So let us now attempt to choose $\xi_\lambda$ so that \eqref{eq:gproof} are satisfied.

Differentiating in $\lambda$ shows that \eqref{eq:gproof} will be satisfied if
\begin{align}
 \L_{\xi_\lambda}\psi &=0\label{gproof_cc}\\
 \L_{\xi_\lambda}\beta_\lambda + \beta_*-\beta & = 0,\label{gproof_dd}
\end{align}
for all $\lambda\in[0,1]$. Since $\beta_\lambda$ is closed \eqref{gproof_dd} will be satisfied in turn if
\begin{equation}\label{gproof_ee}
    \iota_{\xi_\lambda}\beta_\lambda + \alpha_*-\alpha = d\sigma,
\end{equation}
where $\sigma$ is any smooth function. The following argument will identify a $\xi_\lambda$ and a $\sigma$ such that \eqref{gproof_ee} is satisfied. To complete the proof, we will then verify that this $\xi_\lambda$ also satisfies \eqref{gproof_cc}, and that the corresponding flow map $\varphi_{\lambda}$ exists.

In order for a smooth $\xi_\lambda$ to satisfy \eqref{gproof_ee}, it is sufficient that two conditions be met:
\begin{itemize}
    \item[(A)] $\beta_\lambda$ has maximal rank for each $\lambda\in[0,1]$;
    \item[(B)] the one-form $d\sigma + \alpha - \alpha_*$ is in the image of $\beta_\lambda$.
\end{itemize}
Condition (A) is satisfied in a sufficiently-thin neighborhood of $\gamma$ since $\beta$ and $\beta_*$ agree on $\gamma$. Condition (B) may be rephrased as follows: for each $\lambda\in[0,1]$ and each $v\in\text{ker}\beta_\lambda$ we must have $\iota_v(d\sigma +\alpha-\alpha_*)=0$. The following argument shows the latter condition will be satisfied if 
\begin{equation}\label{eq:od_sys}
    \begin{aligned}
         \iota_{\partial_\theta}d\sigma +\iota_{\partial_\theta}(\alpha-\alpha_*)&=0\\
         \iota_{\partial_\zeta} d\sigma + \iota_{\partial_\zeta}(\alpha-\alpha_*)&=0.
    \end{aligned}
\end{equation}


The point is, since $\beta_*$ defined above satisfies $d\psi\wedge \beta_* = 0$ and $d\psi\wedge\beta = d\psi\wedge(\iota_B\Omega) = (\iota_Bd\psi)\Omega = 0$, we have
\[
    d\psi\wedge \beta_\lambda = [1-\lambda]d\psi\wedge \beta + \lambda\,d\psi\wedge \beta_*
            =0.
\]
Therefore if $v$ is in the kernel of $\beta_\lambda$ we must have
\[
    \iota_{v}(d\psi\wedge \beta_\lambda) = (\iota_vd\psi)\beta_\lambda = 0,
\]
which implies $\L_v\psi = 0$ since $\beta_\lambda$ has maximal rank near $\gamma$.
Thus, any vector $v$ in the kernel of $\beta_\lambda$ must be tangent to $\psi$-surfaces in a sufficiently thin neighborhood of $\gamma$. Since the vector fields $\partial_\theta$ and $\partial_\zeta$ span the tangent spaces to the $\psi$-surfaces, a sufficient condition to ensure solvability of \eqref{gproof_ee} for $\xi_\lambda$ is therefore \eqref{eq:od_sys}.


If we can show that there is a smooth $\sigma$ that satisfies the over-determined system 
\eqref{eq:od_sys} then a smooth solution $\xi_\lambda$ of \eqref{gproof_ee} must exist in a sufficiently thin neighborhood of the singular orbit. Moreover, since \eqref{eq:od_sys} eliminates the angular components of $d\sigma +\alpha-\alpha_*$, such a $\sigma$ would satisfy 
\[
    d\psi\wedge(d\sigma +\alpha-\alpha_*)=0.
\]
The wedge product of \eqref{gproof_ee} with $d\psi$ therefore implies $ 0= d\psi\wedge \iota_{\xi_\lambda}\beta_\lambda = (\L_{\xi_\lambda}\psi)\beta_\lambda$, which says that the solution $\xi_\lambda$ would automatically satisfy \eqref{gproof_cc}. Thus, we have reduced the proof to the problem of showing that there is a solution $\sigma$ of \eqref{eq:od_sys}, and that the flow map $\varphi_{\lambda}$ exists.

To establish existence of a smooth solution $\sigma$, we proceed geometrically by analyzing the properties of the one-form $\alpha-\alpha_*$ when pulled back to a $\psi$-level $\Lambda_\psi$. Let $i_\psi:\Lambda_\psi\rightarrow D^2\times \T$ be the natural inclusion map. Then the exterior derivative of the pullback $i_\psi^*(\alpha-\alpha_*)$ is given by
\[
    di_\psi^*(\alpha-\alpha*)  =i_\psi^*\beta - i_\psi^*\beta_* =0,
\]
since $\Lambda_\psi$ is isotropic for $\beta$ (by integrability of $B$) and for $\beta_*$ (by construction). Therefore $i_\psi^*(\alpha-\alpha*)$ is closed for each $\psi$. Moreover, since the fluxes \eqref{eq:actionIntegral} and \eqref{poloidal_flux_def} for $\alpha$ and $\alpha_*$ agree exactly, the periods of $i_\psi^*(\alpha-\alpha*)$ vanish for each $\psi$. It follows that there is a smooth function $\sigma_\psi:\Lambda_\psi\rightarrow \mathbb{R}$ such that $i_\psi^*(\alpha-\alpha*) = -d\sigma _\psi$. Because $\sigma_\psi$ is only defined up to a $\psi$-dependent constant, and the foliation of $D^2\times \T$ by $\psi$-levels admits a smooth transverse section, there is a smooth function $\sigma:D^2\times \T\rightarrow\mathbb{R}$ such that $i_\psi^*\sigma = \sigma_\psi$ for each $\psi$. Thus, if $m=(\widetilde{x},\widetilde{y},\widetilde{\phi})$ is a point in $\Lambda_\psi$ then
\begin{align*}
     \iota_{\partial_\theta}d\sigma 
    & = d[\sigma_\psi](\partial_\theta)\\
    & = -i_\psi^*(\alpha-\alpha_*)(\partial_\theta)\\
    & = -(\alpha-\alpha_*)(\partial_\theta),
\end{align*}
which says that $\sigma $ satisfies the first differential equation of \eqref{eq:od_sys}. A similar calculation shows that $\sigma$ satisfies the second equation of \eqref{eq:od_sys} as well. It follows that a smooth solution of the over-determined system \eqref{eq:od_sys} exists. This shows that there is a smooth $\xi_\lambda$ with smooth $\lambda$-dependence that satisfies both \eqref{gproof_cc} and \eqref{gproof_dd}.

To complete the proof, we will now argue that the flow map $\varphi_\lambda$ exists for $\lambda\in[0,1]$. First we recall some basic properties of the time-dependent vector field $\xi_\lambda$ constructed above. At several steps in our construction, we were forced to shrink the spatial domain on which $\xi_\lambda$ is defined. The end result is that there is some $r>0$, independent of $\lambda$, such that $\xi_\lambda(\widetilde{x},\widetilde{y},\widetilde{\phi})$ is a well-defined smooth function of $(\widetilde{x},\widetilde{y},\widetilde{\phi})$ for each $\lambda$ and  $(\widetilde{x},\widetilde{y},\widetilde{\phi})\in C_r$, where
\begin{align*}
    C_r = \{(\widetilde{x},\widetilde{y},\widetilde{\phi})\in D^2\times \mathbb{T}\mid \sqrt{\widetilde{x}^2 + \widetilde{y}^2}< r\}.
\end{align*}
Our argument also shows that $\xi_\lambda(\widetilde{x},\widetilde{y},\widetilde{\phi})$ is a smooth function of $\lambda$ for $\lambda\in [0,1]$. In fact, our construction would have worked just as well with $\lambda$ in some open set $\mathcal{I}$ containing $[0,1]$. Thus, we may regard $\xi_\lambda(\widetilde{x},\widetilde{y},\widetilde{\phi})$ as a smooth function on the open set $C_r\times \mathcal{I}$.

The above remarks imply that we may define a smooth autonomous vector field $\widetilde{\xi}$ on $C_r\times \mathcal{I}$ using the formula
\begin{align*}
\widetilde{\xi}_\psi(\widetilde{x},\widetilde{y},\widetilde{\phi},\lambda) = \xi_\lambda(\widetilde{x},\widetilde{y},\widetilde{\phi}) + \partial_\lambda.
\end{align*}
We will infer existence of the flow $\varphi_\lambda$ for $\lambda\in[0,1]$ by applying existence and uniqueness theory for autonomous ordinary differential equations (ODEs) to $\widetilde{\xi}$. This is made possible by the following basic fact: the parameterized curve $\widetilde{m}(t) = (\widetilde{x}(t),\widetilde{y}(t),\widetilde{\phi}(t),\lambda(t))$ with $\lambda(0) = 0$ is an integral curve for $\widetilde{\xi}$ if and only if $m(\lambda) = (\widetilde{x}(\lambda),\widetilde{y}(\lambda),\widetilde{\phi}(\lambda))$ is an integral curve for $\xi_\lambda$. 

Now suppose $\widetilde{m}(t) = (\widetilde{x}(t),\widetilde{y}(t),\widetilde{\phi}(t),\lambda(t))$ is an integral curve for $\widetilde{\xi}$ with $\lambda(0) = 0$ and $\widetilde{x}^2(0)+\widetilde{y}^2(0)\leq r_0$ for some $0<r_0<r$. Let $(a,b)\ni t$, $a<0<b$, be the corresponding maximal existence time interval. Since $\widetilde{\xi}$ is a locally Lipschitz function on the open set $C_r\times\mathcal{I}$, Thm. 3.35 in \cite{Meiss17a} implies that if $b$ is finite and $K\subset C_r\times\mathcal{I}$ is compact then $\widetilde{m}(t^*)$ must be outside of $K$ for some $t^*\in[0,b)$. In particular, if we set $K = \overline{C}_{r_0}\times [0,1]$, then we have $\widetilde{m}(t^*) = (\widetilde{x}(t^*),\widetilde{y}(t^*),\widetilde{\phi}(t^*),t^*)$ outside of $\overline{C}_{r_0}\times[0,1]$. (Here we have integrated the equation $\dot{\lambda} = 1$ explicitly.) Since \eqref{gproof_cc} says $\xi_\lambda$ is tangent to level sets of $\psi$, and the boundary of $C_{r_0}$ is a $\psi$ level set, we must have $(\widetilde{x}(t^*),\widetilde{y}(t^*),\widetilde{\phi}(t^*))\in \overline{C}_{r_0}$. Thus, $t^*>1$. Since $b>t^*$, we conclude that the maximal existence time interval for any $\widetilde{\xi}$-integral curve $\widetilde{m}(t)$ with $\widetilde{m}(0)\in \overline{C}_{r_0}\times\{0\}$ must include the closed interval $[0,1]$. The above remarks imply this is equivalent to existence of integral curves $m(\lambda) = (\widetilde{x}(\lambda),\widetilde{y}(\lambda),\widetilde{\phi}(\lambda))$ of $\xi_\lambda$ with $m(0)\in \overline{C}_{r_0}$ for $\lambda\in[0,1]$. By standard uniqueness results for solutions of ODEs, the flow map $\varphi_\lambda:\overline{C}_{r_0}\rightarrow\overline{C}_{r_0}$ exists for $\lambda\in [0,1]$. By standard results on smoothness of flow maps, we also conclude that $\varphi_\lambda$ is a smooth diffeomorphism with smooth $\lambda$-dependence, as desired.
\end{proof}

\subsection{The hyperbolic case: a coisotropic embedding}\label{sec:hyperbolicproof}

    Here we prove \cref{thm:SingularCoordinates} for both direct and reflection 
    hyperbolic singular orbits. In the process, we give a second proof for the elliptic case. The 
    core idea is to embed the 3D presymplectic manifold $(M,\beta)$ into a 4D symplectic manifold
    $(\tilde{M},\omega)$. It is then shown that, in a tubular neighborhood $U$ of a singular orbit $\gamma$ of an integrable system $(B,J,p,\Omega)$, the system can be embedded in an integrable Hamiltonian system with Hamiltonian $H$, and an invariant $\tilde{p}$ on $(\tilde{M},\omega)$ such that $H = 0$ corresponds to the neighborhood $U$ and $\tilde{p}|_U = p$. 
     This embedding allows us to use standard results for action-angle variables of integrable, Hamiltonian systems. In particular, we will use \cref{thm:Eliasson}, below, to show the existence of action-angle variables in a tubular neighborhood $\tilde{U}$ of a two-dimensional surface $\tilde{\gamma}\cong \T\times\R$ that contains $\gamma$ on $H=0$. By restricting back to $H = 0$, we obtain the desired action-angle variables in $U$.
    
    The idea of embedding a presymplectic system in a symplectic system originates with Gotay \cite{gotayCoisotropicImbeddingsPresymplectic1982}. In fact, Gotay demanded that the embedding be \emph{coisotropic}. Recall that a submanifold $M$ of a symplectic manifold $(\tilde{M},\omega)$ is \textit{coisotropic} if, for all $z\in M$,
    whenever $\omega(v_z,w_z) = 0$, for all $v_z \in T_zM$, then $w_z \in T_zM$ as well. Equivalently, one says a submanifold $M$ is coisotropic if 
    the $\omega$-orthogonal complement to $TM$ is a subset of $TM$, that is, $TM^{\perp} \subseteq TM$.
    
    \begin{definition}
        A \emph{coisotropic embedding} of a presymplectic manifold $(M,\beta)$ into a symplectic manifold $(N,\omega)$ is a closed embedding $\pi:M \to N$ such that 
        \begin{enumerate}
            \item $ \pi^* \omega = \beta$,
            \item $TM^\perp \subseteq TM$.
        \end{enumerate}
    \end{definition}

    It has been shown that every presymplectic manifold can be coisotropically embedded in a symplectic manifold. The key concept is the Reeb bundle, sometimes referred to as the characteristic bundle, of $(M,\beta)$:
    \[ 
        E := \cup_{z\in M} E_z,\quad E_z := \set{X \in T_z M | \iota_X \beta_z = 0}.
    \]
    Let $E^*$ denote the dual space of the Reeb bundle.
    
\begin{thm}[Gotay \cite{gotayCoisotropicImbeddingsPresymplectic1982}]\label{thm:GotayEmbedding}
        Let $E$ be the Reeb bundle over $M$. There exists a symplectic form $\omega$ on a tubular neighborhood of the zero section of $E^*$ such that $(E^*,\omega)$ is a symplectic manifold and $(M,\beta)$ is coisotropically embedded in this neighborhood as the zero-section of $E^*$.

        Furthermore, the embedding is unique in that any other such coisotropic embedding is symplectomorphic in this neighborhood. 
\end{thm}
 
    Let us sketch the proof. Embed $M$ as the zero section of $E^*$. The idea for the proof is to make $T_M E^*$ into a symplectic space and use the Weinstein extension theorem \cite{weinsteinLecturesSymplecticManifolds1977a} to get a symplectic neighborhood of this zero section. Then, $M$ is naturally diffeomorphic to the zero section of this bundle.
    
    For each $z\in M$, there is a decomposition $T_z E^* = T_z M \oplus E^*_z$ and since $E$ is a subbundle of $TM$ we can make $T_z M = E_z\oplus G_z$ with $G_z$ some complement. Hence, for each $z\in M$, we have $T_z E^* = G_z \oplus E_z \oplus E^*_z$.

    Let $\pi_1 : T_z E^* \to E_z \oplus E^*_z$ for each $z\in M$ and take $\omega_E$ at $z$ as the canonical symplectic form on $E\times E^*$. Define by $\pi_2$ the bundle map $\pi_2:E\to M$. Then \[\omega_M = \pi_2^*\beta + \pi_1^*\omega_E \] is a symplectic form on $T_ME^*$ and the Weinstein extension theorem  guarantees we can extend $\omega_M$ to a symplectic form $\omega$ on a tubular neighborhood of the zero section in $E^*$. 
    
    Gotay's construction is general, giving the existence of the coisotropic embedding for any presymplectic system. However, with a little more structure, a more concrete coisotropic embedding is evident. 
    
 \begin{thm}[$\eta$-embedding]\label{thm:eta-embedding}
        Suppose there exists a one-form $\eta$ on $M$ such that $\beta\wedge \eta$ is a volume form. Let $\tilde{M} = M\times \R \ni (z,u)$, define the projection $\pi:\tilde{M} \to M$ and take $\tilde{\beta} = \pi^*\beta,\, \tilde{\eta} = \pi^*\eta$. Then, $(\tilde{M},\omega)$ is a symplectic system with
        \begin{equation}\label{eq:omegaDefinition}
            \omega = \tilde{\beta} + d(u \tilde{\eta}).
        \end{equation}
        Moreover, $M$ is coisotropically embedded as $u = 0$ in $
        \tilde{M}$.
\end{thm}

\begin{proof}
        In order for $\omega$ to be symplectic it must be closed and nondegenerate. The first condition is clear as $\beta$ is closed by assumption---implying that $\tilde{\beta}$ is closed, and $d(u \tilde{\eta})$ is exact, hence closed. For nondegeneracy to hold, we need that $\omega\wedge \omega \neq 0$. A calculation gives,
        \[ \omega \wedge \omega = \tilde{\beta}\wedge du 
        \wedge\tilde{\eta} + u\tilde{\beta}\wedge d\tilde{\eta}.  \]
        Along $u = 0$ we have $\omega\wedge \omega = -\beta\wedge\eta\wedge du$. By assumption $\beta\wedge \eta$ is a volume form on $u = 0$, and thus $\beta\wedge\eta\wedge du$ is a nondegenerate form on $u=0$. It follows that $\omega$ is nondegenerate in a tubular neighborhood of $u = 0$. Thus, $(\tilde{M}, \omega)$ is a symplectic form. 
        
        Finally, we show that the embedding is coisotropic. To prove this, it must be shown that $\partial_u$ is not in the $\omega$-orthogonal complement to $TM$. That is, for each $\tilde{z} = (z,0)$ we need to find a vector $v=(\dot{z},\dot{u}) $ tangent to $M$ such that $\omega_{\tilde{z}}(\partial_u,v)\neq 0$. Firstly, on $M$ we have $\beta\wedge \eta = f\Omega$, for some non-vanishing function $f$. It follows that,
        \[ (\iota_B\eta)\beta = f\beta \implies \iota_B\eta = f. \]
        With $\dot{z} = B$  and $\dot{u}=0$ we can conclude that $\omega_{\tilde{z}}(\partial_u,v) = f $, thus, non-vanishing as desired.
\end{proof}

    \begin{remark}
        Assuming that $\beta\wedge \eta$ is a volume form ensures that $M$ is orientable. In turn, this forces the vector bundle $E^*$ to be trivializable, that is, $E^* = M\times \R$. The connection between the constructive embedding of \cref{thm:eta-embedding} and the general version \cref{thm:GotayEmbedding} is now clear: $\tilde{M}$ is diffeomorphic to a tubular neighborhood of the zero-section of $E^*$.
    \end{remark}
    
    On a three-manifold $M$, the desired $\eta$ is always available, for example, one could take $\eta = B^\flat = \iota_B g$ for any Riemannian metric $g$ on $M$. We will call the particular coisotropic embedding of \cref{thm:eta-embedding} the \emph{$\eta$-embedding of $(M,\beta)$}.

    For our purposes, we not only seek a coisotropic embedding of $(M,\beta)$ in some symplectic $(\tilde{M},\omega)$, but we also desire an embedding of the integrable presymplectic system $(B,J,p,\Omega)$ into a symplectic integrable system $(\tilde{B},\tilde{J},H,\tilde{p})$. The following theorem reveals this possibility provided $\eta$ satisfies an additional constraint.
    
 \begin{thm}\label{thm:integrableEmbedding}
        Let $(M,\beta)$ be an orientable, presymplectic manifold and suppose that the associated magnetic field $B$ is integrable with Hamiltonian pair $(J,p)$. Assume there exists a one-form $\eta$ on $M$ such that $\beta\wedge \eta$ is a volume form and that the pullback of $d\eta$ to constant $p$ surfaces vanishes. Let $(M\times\R,\omega)$ be the symplectic manifold obtained from the $\eta$-embedding of $(M,\beta)$, and take coordinates $(z,u)\in M\times\R$. Then $H =  u$ and $\tilde{p}$ form an integrable system on $M\times\R$ with associated Hamiltonian vector fields respectively satisfying
            \begin{equation}
                X_H|_{u=0} = \frac{1}{\iota_{B}\eta}B,\quad X_{\tilde{p}}|_{u=0} =  J - \frac{\iota_{J}\eta}{\iota_{B} \eta} B.
            \end{equation}
\end{thm}
\begin{proof}
        We begin by computing the Hamiltonian vector fields associated to $H$ and $\tilde{p}$ on $u=0$. They are defined through Hamilton's equation. Explicitly we have,
        \[ \iota_{X_H}\omega = - du,\quad \iota_{X_{\tilde{p}}}\omega = -d\tilde{p}. \]
        A calculation reveals that $X_H|_{u=0} = \frac{1}{\iota_{B}\eta}B$ and $X_{\tilde{p}}|_{u=0}=  J - \frac{\iota_{J}\eta}{\iota_{B} \eta} B$, as desired. 
        
        To show that $(H,\tilde{p})$ form a Hamiltonian integrable system on $M\times\R$ we require the Poisson bracket $\{\tilde{p},H\} = 0$. A quick way to see this is to pull back $\omega$ to the common level set $(p,H) = (p_0,H_0)$. The condition $H = H_0$ implies $u = H_0$. The pullback of $\omega$ to the level set is therefore
        \[ 
              \pi_{p_0,H_0}^* (\beta + d(u \eta)) = 0 +\pi_{p_0}^*d(H_0 \eta) = H_0\pi_{p_0^*}d\eta = 0,
        \]
        since $\beta$ and $d(\eta)$ both vanish when pulled back to a $p$-surface.
\end{proof}
    
\begin{remark}
        If a closed $\eta$ can be found then immediately $d\eta$ vanishes on constant $p$ surfaces. A manifold $M$ with a presymplectic form $\beta$ and a form $\eta$ such that $\beta\wedge \eta$ is a volume form is an \emph{almost cosymplectic manifold}. If $\eta$ is in addition closed then $(M,\beta,\eta)$ is a \emph{cosymplectic manifold}. See \cref{sec:Weak} for more details. The above \cref{thm:eta-embedding} is not entirely novel, as work in the embedding of cosymplectic systems has already been developed; see \cite[Lem.~3.2]{leonCosymplecticReductionSingular1993}. However, the theorem generalizes the construction to almost cosymplectic systems.
\end{remark}
\begin{remark}
        Recall that a magnetic field satisfying the MHS conditions implies that $j = \iota_J\Omega = dB^\flat$. Taking $\eta = B^\flat$ we see that the hypothesis of \cref{thm:integrableEmbedding} is satisfied. It follows that magnetic fields satisfying the MHS conditions can be globally embedded in a Hamiltonian integrable system. 
\end{remark}
    
    
    
    
    We now show that a tubular neighborhood $U$ of $\gamma$ can be $\eta$-embedded in a symplectic system. This is a local result. Whether any integrable presymplectic system can be globally embedded is not currently known and deserves further investigation.
    
\begin{corollary}\label{cor:tubularCoIsoEmbedding}
        There exists a tubular neighborhood $U$ of $\gamma$ that can be $\eta$-embedded in $U\times \R \ni(z,u)$ such that $(u,\tilde{p})$ is an integrable system on $U\times\R$.
\end{corollary}
\begin{proof}
        Let $\phi\in \T$ parameterize the curve $\gamma \in M$. As $B$ is non-vanishing and tangent to $\gamma$ we must have $\iota_B d\phi$ be sign-definite on $\gamma$. It follows from the smoothness of $B$ that $\gamma$ is sign definite in a tubular neighborhood $U$ of $\gamma$. Hence, $\beta\wedge d\phi$ is a volume form on $U$. By setting $\eta = d\phi$ and invoking \cref{thm:integrableEmbedding}, the result follows.
\end{proof}

    With \cref{cor:tubularCoIsoEmbedding} established, we have opened the problem of proving \cref{thm:SingularCoordinates} to the theory developed for integrable, Hamiltonian systems. In particular, the large body of work on normal forms in the neighborhood of singular orbits of these systems. We will need to concentrate on four-dimensional, integrable Hamiltonian systems with integrals $F = (F_1,F_2)$. Moreover, the embedded singular orbit $\tilde{\gamma}$ satisfies $d\tilde{p}|_{\tilde{\gamma}} = 0$ but $du|_{\tilde{\gamma}} \neq 0$. Hence, our interest is in 4D systems with a 2D singular orbit $\tilde{\gamma}$ generated from the Hamiltonians $(F_1,F_2)$. The following celebrated theorem gives the analog of \cref{thm:SingularCoordinates} for these 4D systems.
    
\begin{thm}[Eliasson \cite{Eliasson90}]\label{thm:Eliasson}
        Suppose that $F = (F_1,F_2)$ is a smooth (analytic) integrable system on a symplectic 4-manifold $(\tilde{M},\omega)$, and
        let $\gamma \in \tilde{M}$ be a nondegenerate, 2D
        singular orbit of elliptic or hyperbolic type. Then, in a tubular neighborhood $U$ of $\tilde\gamma$, there exists a smooth (analytic) map $\Phi:\R^2 \times \R\times\T \to U$ and coordinates $ (x,y,\Psi,\phi)\in \R^2\times\R\times\T$ such that,
        \begin{enumerate}
            \item $\tilde\gamma$ is given by $x = y = 0$. 
            \item $\Phi^* \omega = dy \wedge dx - d\Psi \wedge d\phi$
            \item $\Phi^*F = G(\psi, \Psi)$, with $G:\R^2 \to \R^2$ a smooth (analytic) diffeomorphism and
            $\psi = \tfrac12(x^2+\varepsilon y^2)$ with $\varepsilon = 1$ if $\tilde\gamma$ is elliptic or $-1$ if hyperbolic. 
        \end{enumerate}
\end{thm}

    The theorem is attributed to Eliasson, however, the analytic case is due to Vey \cite{veyCertainsSystemesDynamiques1978}, and no one reference appears to contain a complete proof of the most general form of \cref{thm:Eliasson}. A brief review is given in \cite{sepeIntegrableSystemsSymmetries2018} and \cite{zungConceptualApproachProblem2018}.
    
    With both the coisotropic embedding and normal form for 4d integrable Hamiltonian systems established, we are now in a position to prove the main theorem, \cref{thm:SingularCoordinates}.
    
\begin{proof}[Proof of \cref{thm:SingularCoordinates}]
        From \cref{cor:tubularCoIsoEmbedding} we can coisotropically embed a neighborhood $U$ of $\gamma\in M$ into $M\times\R\ni (z,u)$ with a symplectic form $\omega$ such that $(u,\tilde{p})$ is an integrable system on $M\times \R$. As $\tilde{p}$ is the trivial extension of $p$, it is clear that the manifold $\tilde{\gamma} = \{ p = 0, u\in \R\}$ is again a nondegenerate singular orbit of the same type (elliptic or hyperbolic) as $\gamma$. 
        By \cref{thm:Eliasson} we are guaranteed a diffeomorphism $\Phi:U\times\R \to \R^2 \times \R \times \T \ni (x,y,\Psi,\phi)$ that is smooth or analytic if $\tilde{p}$ is respectively smooth or analytic, and is such that $\Phi^*\omega = dy \wedge dx - d\Psi \wedge d\phi$. Moreover, there exists a smooth (analytic) diffeomorphism $G:\R^2\to \R^2$ such that $u = G_1(\psi,\Psi),\, \tilde{p} = G_2(\psi,\Psi)$ where $\psi = \tfrac12(x^2+ \varepsilon y^2)$ and $\varepsilon = \pm 1$ depending on whether $\gamma$ is elliptic or hyperbolic.
        
        As $M $ is embedded in $U\times\R$ as $u = 0$ then, taking the inverse of $G$, we have $\Psi|_{M} = G_1^{-1}(0,p)$. As $\gamma$ is nondegenerate, then we are guaranteed an inverse to $\psi = G^{-1}_2(0,p)$ so that $\tilde{p}|_M = p(\psi)$ for $p:\R\to\R$ smooth (analytic). Hence, we can write $\Psi = \Psi(\psi),\, p = p(\psi)$ on $u = 0$. By pulling back $\omega$ to $M = \{u=0\}$ we obtain $\beta$ in the form $\beta = dy\wedge dx - d\Psi(\psi)\wedge d\phi$, as desired.
\end{proof}

\section{Applications: Hamada and Boozer Coordinates}\label{sec:Applications}
As noted in \S\ref{sec:Summary}, two of the most common, straight-field-line (or normal form) coordinate systems that appear in the plasma physics literature are \emph{Hamada} and \emph{Boozer coordinates} \cite{Hamada62, Boozer81,Helander14}. Each of these coordinate systems is defined, in general, only in an open neighborhood of a regular invariant torus. Indeed, in \S\ref{sec:Summary} it was demonstrated that the coordinated can be ill-defined in any neighborhood of an elliptic magnetic axis because the poloidal angle variable degenerates on the axis. In \S\ref{sec:nearAxis}, we will define ``near-axis" Hamada and Boozer coordinates; these  share the essential qualitative features of their well-established counterparts, but explicitly include the magnetic axis within their domain of definition. We will prove the existence of these special coordinate systems under appropriate hypotheses.

We begin, however, by reformulating the regular case for presymplectic systems.
\subsection{Near-surface coordinates}\label{sec:nearTorus}

Recall that the standard definitions, \eqref{eq:HamadaCoord} for Hamada coordinates and \eqref{eq:BoozerCoord} for Boozer coordinates, apply to the neighborhood of a toroidal invariant surface. In the spirit of \S\ref{sec:Presymplectic}, we extend these definitions more generally to integrable presymplectic systems.
\begin{definition}[Near-surface Hamada coordinates]
Suppose $M$ is a compact three-manifold and that $(B,J,p, \Omega)$ is an integrable presymplectic system, \cref{def:integrable}.
Let $\Lambda$ be a regular invariant torus with a tubular neighborhood $\U$. A diffeomorphism $\Phi_H:\U\rightarrow  (\psi_1,\psi_2) \times \T\times \T$, with $\Phi_H(m) = (\psi,\theta,\zeta)$, defines a system of \emph{Hamada coordinates} in the flux interval $\psi_1<\psi < \psi_2$ if
\begin{equation}\label{eq:HamdaForm}
\begin{split}
    \beta = \iota_B\Omega& = dF\wedge d\theta - dG\wedge d\zeta,\\
    j = \iota_J\Omega & = dK\wedge d\theta - dL\wedge d\zeta,
\end{split}
\end{equation}
where $F,G,K,L$ are smooth single-variable functions of $\psi$. 
\end{definition}

By contrast, it is unclear what are the precise requirements for the existence of Boozer coordinates, given in the standard case by \eqref{eq:BoozerCoord}, but certainly they can be defined in the case that $J$ is the MHS current. The extension to presymplectic systems becomes:
\begin{definition}[Near-surface Boozer coordinates]
Suppose $M$ is a compact three-manifold with metric $g$ and that $(B,J,p, \Omega)$ is an MHS integrable system, \cref{def:MHSIntegrable}. A diffeomorphism $\Phi_B:\U\rightarrow (\psi_1,\psi_2)\times\T\times \T$, with $\Phi_B(m) = (\psi,\theta,\zeta)$, defines a system of \emph{Boozer coordinates} in the flux interval $\psi_1<\psi < \psi_2$ if
\begin{equation}\label{eq:BoozerForm}
\begin{split}
    \iota_B\Omega &= df\wedge d\theta - dg\wedge d\zeta,\\
         \iota_Bg\wedge dp & = dk\wedge d\theta - dl\wedge d\zeta,
\end{split}
\end{equation}
where $f,g,k,l$ are smooth single-variable functions of $\psi$. 
\end{definition}

For the sake of completeness, we summarize the existence theory for these coordinate systems.

\begin{thm}[Existence of Regular Hamada and Boozer coordinates]\label{thm:RegularHamada}
If $\Lambda$ is a regular invariant torus for an integrable presymplectic system then there is a tubular neighborhood of $\Lambda$ on which there is system of Hamada coordinates. If the system is in addition MHS integrable, then there is a neighborhood of $\Lambda$ on which there is a system of Boozer coordinates.
\end{thm}
\begin{proof}
The proof for the Hamada case using the Arnol'd-Liouville argument is given in \cite{Burby_Kallinikos_MacKay_2020a}. The proof for the Boozer case is the same as the proof of \cref{NAB_exist}.
\end{proof}

Of course, these coordinates are not unique since the angles on a torus do not have a unique definition.
\begin{thm}[Characterization of Hamada coordinates]\label{hamada_characterized}
Different sets of Hamada coordinates on a tubular neighborhood $\U$ of a regular torus $\Lambda$ with the same flux variable $\psi$ are equivalent up to an affine, unimodular transformation on each flux surface. 

Specifically,  $(\psi, \theta^{(1)}, \zeta^{(1)})$ and $(\psi, \theta^{(2)}, \zeta^{(2)})$
are Hamada coordinates on $\U$ {iff} there is a constant matrix $A\in SL(2,\mathbb{Z})$ and
a pair of smooth flux functions $\nu_1(\psi),\nu_2(\psi)$ such that
\begin{align}\label{angle_transformation}
    \begin{pmatrix} \theta^{(2)}\\ \zeta^{(2)} \end{pmatrix}
    = A
     \begin{pmatrix}\theta^{(1)}\\\zeta^{(1)} \end{pmatrix} 
     +\begin{pmatrix} \nu_1\\ \nu_2 \end{pmatrix} ,
\end{align}
and so that the coefficients of the forms \eqref{eq:HamdaForm} are related by
\begin{align}\label{eq:flux_fun_relation}
    \begin{pmatrix}
        F^{(2)} & -G^{(2)}\\
        K^{(2)} & - L^{(2)}
    \end{pmatrix}
    = \begin{pmatrix}
        F^{(1)} & -G^{(1)}\\
        K^{(1)} & - L^{(1)}
    \end{pmatrix} A^{-1}.
\end{align}
\end{thm}
\begin{proof}
Verification of the \textit{only if} statement amounts to a straightforward calculation, and so we only prove the forward direction. 

Since the second cohomology of $\U$ is trivial, the two-forms $\beta = \iota_B\Omega$ and $j = \iota_J\Omega$ each possess a pair of primitives, $\beta = d\alpha^{(1)} = d\alpha^{(2)}$, $j = d\kappa^{(1)} = d\kappa^{(2)}$, where
\begin{align*}
    \alpha^{(k)} &= F^{(k)}(\psi)\,d\theta^{(k)}- G^{(k)}(\psi)\,d\zeta^{(k)}\\
    \kappa^{(k)} &= K^{(k)}(\psi)\,d\theta^{(k)} - L^{(k)}(\psi)\,d\zeta^{(k)}
\end{align*}
for $k = 1,2$. Using matrix notation,
\begin{align*}
    \Pi^{(k)} = \begin{pmatrix}
                \alpha^{(k)}\\ \kappa^{(k)}
               \end{pmatrix},\quad 
    \vartheta^{(k)} 
             = \begin{pmatrix} \theta^{(k)}\\ \zeta^{(k)} \end{pmatrix},\quad 
     \Psi^{(k)} = \begin{pmatrix}
                F^{(k)} & -G^{(k)}\\
                K^{(k)} & - L^{(k)}
              \end{pmatrix},
\end{align*}
these formulas become simply
\begin{equation}\label{primitive_relation}
    \Pi^{(1)} = \Psi^{(1)} d\vartheta^{(1)},\quad \Pi^{(2)} = \Psi^{(2)} d\vartheta^{(2)} .
\end{equation}

Since both $\vartheta^{(1)}$ and $\vartheta^{(2)}$ define pairs of angular variables on $\U$, there exists an $A\in SL(2,\mathbb{Z})$ and a single-valued $\nu = (\nu_1(\psi,\theta^{(1)},\zeta^{(1)}),\nu_2(\psi,\theta^{(1)},\zeta^{(1)}))$ such that
\begin{equation}\label{angle_relation}
    \vartheta^{(2)} = A\vartheta^{(1)} + \nu .
\end{equation}
Substituting \eqref{angle_relation} into \eqref{primitive_relation} therefore implies
\[
    \Pi^{(2)} = \Psi^{(2)} A d\vartheta^{(1)} + \Psi^{(2)} d\nu.
\]
But since $d\Pi^{(1)} = d\Pi^{(2)}$ this implies
\begin{equation} \label{matrix_relation_pf}
\begin{split}
    \Psi^{(1)\prime} d\psi \wedge d\vartheta^{(1)} &= d\Pi^{(1)}\\
            & = d\Pi^{(2)} 
             = \Psi^{(2)\prime} A\,d\psi\wedge d\vartheta^{(1)} + \Psi^{(2)\prime}d\psi\wedge d\nu \\
            & = \left[\Psi^{(2)\prime} A +\Psi^{(2)\prime} D_\vartheta \nu\right]d\psi\wedge d\vartheta^{(1)},
\end{split}
\end{equation}
where $D_\vartheta \nu$ is the Jacobian of $\nu$ with respect to $\vartheta^{(1)}$.
Because $d\psi\wedge d\theta^{(1)}$ and $d\psi\wedge d\zeta^{(1)}$ are linearly-independent,  \eqref{matrix_relation_pf} can only be satisfied if
\begin{equation}\label{pre_average}
    \Psi^{(1)\prime} = \Psi^{(2)\prime} A + \Psi^{(2)\prime} D_\vartheta \nu.
\end{equation}
Averaging  \eqref{pre_average} over $\theta^{(1)},\zeta^{(1)}$ then implies
\[
    \Psi^{(1)\prime} = \Psi^{(2)\prime} A,
\]
which is readily seen to give \eqref{eq:flux_fun_relation}. Substituting \eqref{eq:flux_fun_relation} back into \eqref{pre_average} therefore implies
\[
    0 =\Psi^{(2)\prime} D_\vartheta \nu, 
\]
implying that $D_\vartheta \nu = 0$ because the matrix $\Psi^{(2)}$ is non-singular by the linear independence of $J$ and $B$ on $\U$. It follows that $\nu$ is only a function of $\psi$, and thus \eqref{angle_relation} is equivalent to \eqref{angle_transformation}.
\end{proof}

\begin{remark}
The proof of \cref{NAB_exist} shows that Boozer coordinates may be understood as Hamada coordinates with respect to a different volume form. Therefore the above characterization of Hamada coordinates applies to Boozer coordinates as well.
\end{remark}

\subsection{Near axis Hamada coordinates}\label{sec:nearAxis}

In this section, we extend the definition of Hamada coordinates to the neighborhood of an elliptic magnetic axis, recall \cref{def:ellipticAxis}.

\begin{thm}[Near-axis Hamada coordinates]\label{NAH_existence}
Suppose that $M$ is a three-manifold and  $(B,J,p,\Omega)$ is an integrable presymplectic system. If $\gamma$ is an elliptic magnetic axis and $\U$ is a tubular neighborhood of $\gamma$ then there exists a diffeomorphism $\Phi_{\text{NAH}}:\U\rightarrow D^2\times \T$ with $\Phi_{\text{NAH}}(m) = (x,y,\phi)$, such that 
\begin{align*}
   p &= P(\psi),\\
  \beta &= \iota_B\Omega =  F^\prime(\psi)\,dy\wedge dx - dG\wedge d\zeta\\
    j &= \iota_J\Omega = K^\prime(\psi)\,dy\wedge dx - dL\wedge d\zeta,
\end{align*}
in $\U$, where $F,G,K,L,P$ are smooth single-variable functions of $\psi = \tfrac12(x^2 + y^2)$.
\end{thm}

\begin{remark}
Note that by a rescaling of $(x,y)$ we can always arrange for $F^\prime(\psi) = 1$.
\end{remark}

\begin{proof}
By \cref{thm:SingularCoordinates}, there is a tubular neighborhood $\U$ of $\gamma$ and a $C^\infty$ diffeomorphism $\Phi_0:\U\rightarrow D^2\times T$, with
$\Phi_0(m) = (X,Y,\phi)$, such that
\begin{equation}\label{eq:iotaNormal}
\begin{split}
    p & = P(\psi)\\
    \iota_B\Omega &= dY\wedge dX - G^\prime(\psi)\,d\psi\wedge d\phi,
\end{split}
\end{equation}
where $\psi = \tfrac12(X^2 + Y^2)$ is the toroidal flux and $G = \Psi_P$ in the statement of \cref{thm:SingularCoordinates}.
Note that both $p$ and $\iota_B\Omega$ are invariant under the torus action 
\[
    T_{\theta,\zeta}(X,Y,\phi) = (X\cos\theta -Y\sin\theta,Y\cos\theta+X\sin\theta,\phi+\zeta),
\]
while the two-form $j =\iota_J\Omega$ need not be. 
Since $D^2\times \T$ has trivial second de Rham cohomology, there
must be a one-form $\kappa$ such that $j = d\kappa$. Define the smooth single-variable functions
\[
    K(\psi)  = \frac{1}{2\pi}\oint_{c_P}\kappa, \qquad
    L(\psi) =-\frac{1}{2\pi}\oint_{c_T}\kappa,
\]
where $c_P$ is any closed curve in a constant $\psi$ surface that is homologous to $\theta\mapsto T_{\theta,0}(X,Y,\phi)$ and $c_T$ is any curve homologous to $\zeta\mapsto T_{(0,\zeta)}(X,Y,\phi)$ in the same surface. Associated with these action variables is the one-form
\[
    \kappa_* = \tfrac{K(\psi)}{\psi}\tfrac12(Y\,dX - X\,dY) - 
             L(\psi)\,d\phi ,
\]
and the associated two-form 
\[
    j_* = d\kappa_*
         = K^\prime(\psi)\,dY\wedge dX - L^\prime(\psi)\,d\psi\wedge d\phi.
\]
Note that $K(\psi)/\psi$ is a smooth function of $\psi$ because $K(0) = 0$ and $K$ is smooth. The goal for the remainder of the proof is to find a diffeomorphism $(X,Y,\phi)\mapsto(x,y,\zeta)$ such that $\iota_J\Omega$ is transformed into $j_*$ while $p$ and $\iota_B\Omega$ retain their normal forms \eqref{eq:iotaNormal}. The desired diffeomorphism will be constructed as the $\lambda=1$ flow map for a time-dependent vector field $\xi_\lambda$.

Let $\xi_\lambda$ be a time-dependent vector field on $D^2\times \T$ with flow map $\varphi_\lambda : D^2\times \T\rightarrow D^2\times \T$. 
Defining
\begin{align*}
    j_\lambda &= [1-\lambda]j + \lambda j_*,\\ 
    \kappa_\lambda &= [1-\lambda]\kappa + \lambda\kappa_*,
\end{align*}
suppose that the flow map $\varphi_\lambda $ satisfies
\begin{align*}
   \varphi_\lambda ^*\beta & =  \beta, \\
   \varphi_\lambda ^* j_\lambda  &= j,
\end{align*}
for each $\lambda\in[0,1]$.
By differentiating with respect to $\lambda$, we see that these conditions on $\varphi_\lambda$ are equivalent to
\begin{equation}\label{eq:hproof}
\begin{split}
     \L_{\xi_\lambda}\beta &= d \iota_{\xi_\lambda} \beta = 0\\
     \L_{\xi_\lambda} j_\lambda  + \partial_\lambda j_\lambda &= 
      d (\iota_{\xi_\lambda} j_\lambda + \partial_\lambda \kappa) = 0.
\end{split}
\end{equation}
The first equation in \eqref{eq:hproof} will be satisfied when $\xi_\lambda$ is parallel to $B$, i.e., $\xi_\lambda = b_\lambda B$, for some smooth scalar function $b_\lambda$. The second equation in \eqref{eq:hproof} will also be satisfied if there is a time-dependent scalar function $S_\lambda$ such that
\begin{equation}\label{eq:hproof-c}
      b_\lambda \,\iota_{ B}j_\lambda + \partial_\lambda\kappa_\lambda = dS_\lambda.
\end{equation}
To show that there exist such smooth $b_\lambda$ and $S_\lambda$, we first observe that $\partial_\lambda\kappa_\lambda = \kappa_*-\kappa$. When pulled back to any $\psi$-surface, the one-form $\partial_\lambda\kappa_\lambda$ is closed because $\psi$-surfaces are isotropic for both $j$ (by hypothesis) and $j_*$ (by construction). Moreover, the periods of $\partial_\lambda\kappa_\lambda$ on a $\psi$-surface all vanish since $\kappa$ and $\kappa^*$ have the same action integrals. Thus $\partial_\lambda\kappa_\lambda$ is exact on each $\psi$-surface. It follows that we may define a smooth function $S$ on $D^2\times \T$ such that $dS$ agrees with $\partial_\lambda\kappa_\lambda$ on each $\psi$-surface. In other words, there exists a smooth $S$ such that 
\begin{align}
    \partial_\lambda\kappa_\lambda - dS_\lambda = h_\lambda\,d\psi,
\end{align}
for some smooth function $h_\lambda$. Next let
\begin{align*}
    \partial_\theta &= \partial_\theta T_{\theta,0}|_{\theta=0} = -Y\partial_X +X\partial_Y \\
    \partial_\zeta &= \partial_\zeta T_{0,\zeta} |_{\zeta=0}=\partial_\phi
\end{align*}
be infinitesimal generators for the torus action $T_{\theta,\zeta}$. Also assume, without loss of generality, that $\Omega = \rho\,dX\wedge dY\wedge d\phi$, where $\rho$ is a smooth positive function bounded away from $0$. 

By \eqref{eq:iotaNormal}, the magnetic field $B$ must be of the form
\[
    B  = \frac{\partial_\phi +G^\prime(\psi)\partial_\theta}{\rho}.
\]
Therefore
\begin{align*}
    \iota_{ B}j_\lambda & = [1-\lambda]\iota_{B}j + \lambda \iota_{B}j_*\\
            &=[1-\lambda]dp+\lambda \iota_{B}(K^\prime(\psi)dY\wedge dX - L^\prime(\psi)\,d\psi\wedge d\phi)\\
            &\equiv r_\lambda\,dp.
\end{align*}
where we have defined
\[
    r_\lambda = [1-\lambda] + \lambda \frac{L^\prime(\psi)-G^\prime(\psi)K^\prime(\psi)}{\rho P^\prime(\psi)}.
\]
Suppose that $r_\lambda$ is positive and bounded away from zero for $\lambda\in[0,1]$. Then $b_\lambda = -h/r_\lambda$ is smooth and satisfies \eqref{eq:hproof-c}. Thus, $\Phi=(\varphi_1)^{-1}$ satisfies
\begin{align*}
    \Phi_*p     & = p\\
    \Phi_*\beta &= dY\wedge dX - G^\prime(\psi)\,d\psi\wedge d\zeta\\
    \Phi_*j     &= K^\prime(\psi)\,dY\wedge dX - L^\prime(\psi)\,d\psi\wedge d\zeta,
\end{align*}
and $(x,y,\zeta) = \Phi(X,Y,\phi)$ are smooth near-axis Hamada coordinates. The proof will therefore be complete if we can show that $r_\lambda$ is indeed positive and bounded away from zero.

To examine the sign of $r_\lambda$ near the magnetic axis, we pick a $\psi$-level $M$ with $\psi$ arbitrary but nonzero, and construct standard Hamada coordinates $(\psi,\theta,\zeta)$ in a neighborhood of $M$, where $\psi$ is the $p$-surface toroidal magnetic flux. We assume the Hamada angles $\theta,\zeta$ are chosen so that
\[
    \frac{1}{2\pi}\begin{pmatrix}
                     \oint_{c_P} d\theta & \oint_{c_T}d\theta\\
                    \oint_{c_P} d\zeta & \oint_{c_T} d\zeta
                \end{pmatrix} 
            = \begin{pmatrix}
                    1 & 0 \\
                    0 & 1
              \end{pmatrix},
\]
with poloidal and toroidal loops $c_P$ and $c_T$ as before.
If they are not, then we define $A\in SL(2,\mathbb{Z})$ according to 
\[
        A = 2\pi\begin{pmatrix}
                \oint_{c_P} d\theta & \oint_{c_T}\theta\\
                \oint_{c_P} d\zeta & \oint_{c_T} d\zeta
            \end{pmatrix}^{-1},
\]
and set $(\overline{\theta},\overline{\zeta})^T = A (\theta,\zeta)^T$. According to \cref{hamada_characterized}, the mapping $(\overline{\psi,\theta},\overline{\zeta})$ then defines a system of Hamada coordinates with the desired property.

In these coordinates we have
\begin{align*}
        \Omega &= \tfrac{1}{(2\pi)^2} V^\prime(\psi)\,d\psi\wedge d\theta\wedge d\zeta\\
        \iota_B\Omega & = F^\prime(\psi)\,d\psi\wedge d\theta - G^\prime(\psi)\,d \psi\wedge d\zeta\\
        \iota_J\Omega & = K^\prime(\psi)\,d \psi\wedge d\theta - L^\prime(\psi)\,d\psi\wedge d\zeta,
\end{align*}
where $Vol(\psi)$ is the smooth single-variable function given by
\begin{align*}
    Vol(\psi) = \int_{\tfrac12(X^2+Y^2)\leq \psi}\Omega.
\end{align*}
Therefore the vector fields $B,J$ are given by 
\begin{align*}
        B &= \frac{(2\pi)^2}{V^\prime(\psi)}\left(F^\prime(\psi)\partial_\zeta + \,G^\prime(\psi)\,\partial_\theta\right)\\
        J& = \frac{(2\pi)^2}{V^\prime(\psi)}\left(K^\prime(\psi)\partial_\zeta + \,L^\prime(\psi)\,\partial_\theta\right),
\end{align*}
and the Hamiltonian equation $\iota_J\beta=\iota_J\iota_B\Omega = -dp$ implies
\[
    1 = (2\pi)^2\,\frac{ L^\prime(\psi)- G^\prime(\psi)\,K^\prime(\psi)}{P^\prime(\psi)\,V^\prime(\psi)}.
\]
Away from the magnetic axis, the function $r_\lambda$ may therefore be written
\[
    r_\lambda = [1-\lambda] + \lambda\,\frac{V^\prime(\psi)}{(2\pi)^2\,\rho}.\label{rho_formula}
\]
And by continuity of $r_\lambda$, \eqref{rho_formula} must be valid on the magnetic axis $X=Y=0$ as well. In particular, since
\begin{align*}
    V^\prime(0) &= \frac{d}{d\lambda}\bigg|_{\lambda=0} \iiint_{X^2/2 + Y^2/2\leq \lambda} \rho(X,Y,\phi) \,dX\,dY\,d\phi\\
        & =\frac{d}{d\lambda}\bigg|_{\lambda=0} \int_0^{\lambda}\int_{0}^{2\pi}\int_0^{2\pi} \rho(\sqrt{2\varepsilon}\cos\theta,\sqrt{2\varepsilon}\sin\theta,\phi)\,d\theta\,d\phi\,d\varepsilon\\
        & =2\pi \int_0^{2\pi} \rho(0,0,\phi)\,d\phi,
\end{align*}
the value of $r_\lambda$ on-axis is 
\begin{equation}\label{ineq}
\begin{split}
    r_\lambda(0,0,\phi) &= [1-\lambda]+\lambda\,\frac{V^\prime(0)}{(2\pi)^2\rho(0,0,\phi)}\\
        & =[1-\lambda]+\lambda\,\frac{[2\pi]^{-1}\int_0^{2\pi}\rho(0,0,\overline{\phi})\,d\overline{\phi}}{\rho(0,0,\phi)}\\
        & \geq [1-\lambda] + \lambda\frac{\rho_{\text{min}}}{\rho_{\text{max}}}\\
        &\geq \frac{\rho_{\text{min}}}{\rho_{\text{max}}},
\end{split}
\end{equation}
where $\rho_{\text{min}}$ and $\rho_{\text{max}}$ are the minimum and maximum values of the periodic function $\phi\mapsto \rho(0,0,\phi)$, respectively. By continuity of $r_\lambda$, \eqref{ineq} that $r_\lambda > c$ for some positive constant $c$ in a sufficiently thin neighborhood of the magnetic axis.

\end{proof}

\begin{remark}
Observe that this proof exploits the fact that the normal form of \cref{thm:SingularCoordinates} is unchanged by transformations that stretch along the magnetic field lines.
\end{remark}

\subsection{Near axis Boozer coordinates}
In this section, we extend Boozer coordinates to accommodate an elliptic magnetic axis.

\begin{thm}[Near-axis Boozer coordinates]\label{NAB_exist}
Suppose that $(B,J,p, \Omega)$ is an MHS integrable presymplectic system on a Riemannian three-manifold $(M,g)$. If $\gamma$ is an elliptic magnetic axis, then there exists a tubular neighborhood $\U$ of $\gamma$ and a $C^\infty$ diffeomorphism $\Phi_{\text{NAB}}:\U\rightarrow D^2\times \T$ that defines $C^\infty$
coordinates such that
\begin{align*}
    \iota_B\Omega &= f^\prime(p)\,dy\wedge dx - dg\wedge d\zeta\\
    \iota_Bg\wedge dp & = k^\prime(p)\,dy\wedge dx - dl\wedge d\zeta,
\end{align*}
in $\U$, where $f,g,k,l$ are smooth single-variable functions of $p$.
\end{thm}
\begin{proof}
Introduce the two-forms $\beta =\iota_B\Omega$ and $\mu =  \iota_Bg\wedge dp$.
Notice that level sets of $p$ are isotropic for both $\beta$ and $\mu$. Also notice that both forms are closed because
\begin{align*}
    d \mu&= d( \iota_Bg\wedge dp)
     = dp\wedge d\iota_B g\\
    & = dp\wedge\iota_J\Omega
     =  \L_Jp\,\Omega\\
    &=0,
\end{align*}
where we have used the relation $di_B g = \iota_J \Omega$ satisfied by MHS integrable systems.

Now introduce the volume form $\Omega^\prime = |B|^2\Omega$ where $|B|^2 = g(B,B)$, with $g$ the Riemannian metric on $M$. The associated vector fields $B^\prime,J^\prime$ defined according to  \eqref{eq:Bprime}, satisfy
\[\iota_{B^\prime}\Omega^\prime = \beta\qquad \iota_{J^\prime}\Omega^\prime  = \mu.\]

Note that \cref{lem:BProperties} implies that $B^\prime $ and $J^\prime$ are  divergence-free with respect to $\Omega^\prime$. Moreover, because
\begin{align*}  
     \L_{B^\prime}\mu& = -( \L_{B^\prime}dp)\wedge \iota_Bg - dp\wedge ( \L_{B^\prime}\iota_Bg)\\
        & = -dp\wedge (\iota_{B^\prime}d\iota_Bg + d[g(B,B^\prime)] )\\
        & = -dp\wedge (\iota_{B^\prime}\iota_J\Omega)\\
        & = -dp\wedge(dp)/|B^2|\\
        & = 0,
\end{align*}
$B^\prime$ and $J^\prime$ commute. Lastly, 
\begin{align*}
    \iota_{J^\prime}\beta & = -\iota_{B^\prime}\mu\\
        & = \iota_{B^\prime}(dp\wedge \iota_Bg)\\
        & = -g(B,B^\prime)\,dp\\
        & = -dp,
\end{align*}
which says that $(J^\prime,p)$ is a Hamiltonian pair for the presymplectic form $\beta$. It follows that $(B^\prime,J^\prime,p,\Omega^\prime)$ defines an integrable presymplectic system. Moreover $\gamma$ is an elliptic magnetic axis for $(B^\prime,J^\prime,p,\Omega^\prime)$. The desired result is therefore implied by the existence of near-axis Hamada coordinates (\cref{NAH_existence}) for integrable presymplectic systems.

\end{proof}

\section{Discussion}
The work in this paper formulates and proves results pertinent to the study of plasma confinement by magnetic fields. However, it also attempts to contextualize the results in a more general framework: that of presymplectic systems. In particular, many of our results apply to magnetic fields with flux surfaces that may not satisfy the ideal MHD euqilibrium equation $(\nabla\times B)\times B = \nabla p$. A benefit of relaxing this common assumption is that our results allow for pressure-anisotropic equilibria and equilibria with flow. Due to this balancing between application and theory, sacrifices have been made on both sides, ultimately leaving several worthwhile future directions.

Firstly, the theory throughout has been described for presymplectic systems on an orientable three-manifold. The specialization to three dimensions is purely because of the physical application and many of the results should still be true for presymplectic systems on a $2n+1$-dimensional manifold. In particular, a theorem analogous to \cref{thm:integrableEmbedding} should hold and one should be able to obtain,  using preexisting results from symplectic geometry, an analog to \cref{thm:SingularCoordinates} near singular orbits in higher dimensions. This would include more exotic singularities with a focus-focus component. 

Our results do not apply if the magnetic axis is degenerate, that is, if the normal Hessian has a zero eigenvalue. It is unclear to us which aspects of our results hold in the degenerate case. Indeed, even in the symplectic realm such questions are unanswered. Further work in this area would prove significant to several fields.

It may also be possible to broaden some of the hypotheses throughout. For instance, in the definition of an integrable system for a presymplectic manifold, \cref{def:integrable}, four pieces of information are required; $B,J,p,$ and $\Omega$. A natural question arises of whether this is minimal. That is, it may be possible to show that given any non-vanishing, volume preserving vector field $B$, if $p$ is invariant under $B$, then there exists a volume-preserving vector field $J$ such that $\iota_J \beta = -dp$. If such a result held then only the information $B,p,\Omega$ would be required. It would also be significant if one could show that such a construction of $J$ was not possible; this would then imply that knowing that a magnetic field lies on $p$ surface would not be sufficient for integrability.

Another possible refinement comes from \cref{thm:integrableEmbedding}. The result shows the possibility of embedding an integrable presymplectic system into an integrable system on the well-known guiding center symplectic phase space \cite{Littlejohn_1979}, provided that there is a one-form $\eta$ that satisfies certain restrictions. But are these restrictions are necessary? For example, if it was only required that $\beta\wedge \eta$ be a volume form, then we would have a powerful way of embedding any presymplectic integrable system simply by taking $\eta = B^\flat$. This global embedding would serve as a powerful theoretical tool in future work on existence of coordinates or other structures, such as quasisymmetry \cite{Helander14}. Moreover, it is linked to the previous question of the existence of a global volume preserving $J$ such that $\iota_J \beta = -dp$. For, if one has a global integrable embedding, then $H,\tilde{p}$ will give an integrable system on the symplectic space and the Hamiltonian vector field associated to $\tilde{p}$ should restrict to the desired $J$ on the embedded space.

There also remain unanswered questions for the application of our results to special coordinate systems such as those in \S\ref{sec:Applications}. Crucially, we have not shown whether some version of Hamada or Boozer coordinates exist near a hyperbolic magnetic axis. While we conjecture that this is true, the non-connectedness of level sets in a tubular neighborhood of a hyperbolic axis prevent the approach we have used for the elliptic case. Showing the result for the hyperbolic case is more than simply closing this gap. The existence of Hamada coordinates is equivalent to the existence of a Hamiltonian circle action, or $\T$-symmetry, of the system. This circle action is an invaluable tool in both symplectic and presympletic geometry.

The precise necessary conditions on $B,J,p$ for the existence of Boozer coordinates is yet to be investigated. Certainly, from \cref{NAB_exist}, it is sufficient if $B$ is MHS integrable. However, following the same proof, it is probable to that one can also show existence if $j = f(\psi) dB^\flat$ for any smooth function $f$. Thus, MHS is not necessary. 

Lastly, there are several remaining questions on global existence of coordinates. In particular, are there topological constraints that imply that there are  globally defined Hamada and Boozer coordinates? In addition, particularly important for quasisymmetry, are there constraints that imply the global existence of a Hamiltonian circle action?

\section*{Acknowledgements}
ND and JDM acknowledge support from the Simon's Foundation Grant \#601972, ``Hidden Symmetries and Fusion Energy.'' JWB was supported by the Los Alamos National Laboratory LDRD program under project 20180756PRD4. Useful conversations with Matt Landreman and Gabriel Plunk are gratefully acknowledged. 
\section*{Data Availability}
The data that support the findings of this study are available within the article.

\begin{appendix}

\section{Symplectic, Presymplectic, Co-symplectic and Contact}\label{sec:Weak}
In this appendix we collect some definitions of terms that correspond to weaker concepts than symplectic dynamics. We suppose that the phase space is a manifold $M$, with a volume form $\Omega$. Recall that a \textit{symplectic} form $\omega$, is a two-form, $\omega$, that is closed, $d\omega = 0$,
and nondegenerate. A form is nondegenerate if $\omega(v,w) = 0$ for all $w \in \mathfrak{X}(M)$ only when $v = 0$. A vector field $X_H \in \mathfrak{X}(M)$ is a Hamiltonian vector field when there is a smooth function $H: M \to \mathbb{R}$ such that \eqref{eq:Hamilton} is satisfied.

When the manifold $M$ has an odd dimension, then there is no symplectic structure, since any two-form must be degenerate. We summarize some of the alternative structures that are used in the literature in \cref{fig:structures}.

A closed two-form, $\beta$, that is degenerate is a \textit{presymplectic} form \cite{Ortega04,Gotay79a}. It is usually assumed that $\beta$ has constant rank \cite{Gotay79b}; we will suppose that $\beta$ has maximal rank, that is $\ker{\hat\beta_z}$ \eqref{eq:BetaKernel} has dimension one. 
This case is sometimes called a \textit{contact} form \cite{Abraham78} when $\dim{M} = 2n+1$ is odd, though this does not seem to be universal. 

If in addition, there is a one-form $\eta$ such that $\beta^n \wedge \eta = \Omega$, then $(\beta,\eta)$ is a \textit{cosymplectic} structure. In this case $\ker{\hat{\beta}_z}$ is one-dimensional. When $\beta = d\eta$, this structure becomes an \textit{exact contact} structure \cite{Abraham78}. A Darboux-like theorem implies that in this case there are local coordinates $(q_1,p_1,\ldots, q_n,p_n,,u)$ on $M$ such that $\beta = dq \wedge dp = \sum_{i=1}^n dq_i\wedge dp_i$ and $\eta = du$. The standard example of an exact contact manifold corresponds to a nonautonomous Hamiltonian system with $M = \mathbb{R}^{2n} \times \mathbb{R}$, with $\beta = \omega$ and $\eta = dt$.

\begin{center}
  \footnotesize
  \begin{tikzpicture}[auto,
    block_left/.style ={rectangle, draw=black, thick, fill=white,
      text width=16em, text ragged, minimum height=4em, inner sep=6pt},
    line/.style ={draw, thick, -latex', shorten >=0pt}]
    \matrix [column sep=5mm,row sep=3mm] {
      \node[block_left] (manifold) {
        \textbf{Manifold:} $M^{2n+1}$
        \begin{itemize}
            \item $\beta \in \Lambda^2(M)$
        \end{itemize}
      };
      & \\
      \node [block_left] (presymplectic) {
      	\textbf{presymplectic:} \\
      	\begin{itemize}
      	\item $d\beta = 0$
      	\item $\beta$ of maximal rank
      	\end{itemize}
  		};
      &  \\
       \node [block_left] (almostCosymplectic) {
      	\textbf{almost cosymplectic:} \\
      	$\eta \in\Lambda^1(M)$ s.t.
      	\begin{itemize}
      	\item $\beta^n \wedge \eta$ a volume form
      	\end{itemize}
      };
      & \\
      \node [block_left] (cosymplectic) {
      	\textbf{cosymplectic:} \\
     	\begin{itemize}
      	\item $d\eta = 0$ 
      	\end{itemize}
      }; 
      & \node [block_left] (contact) {
      	\textbf{exact contact:} \\
      	$\eta\in\Lambda^1(M)$ s.t.
      	\begin{itemize}
      	\item $\beta = d\eta$
      	\item $(d\eta)^n \wedge \eta$ a volume form.
      	\end{itemize}
      };\\
    };
    \begin{scope}[every path/.style=line]
      \path (manifold) -- (presymplectic);
      \path (presymplectic) -- (almostCosymplectic);
     \path (almostCosymplectic) -- (cosymplectic);
     \path (presymplectic) -| (contact);
    \end{scope}
  \end{tikzpicture}
    
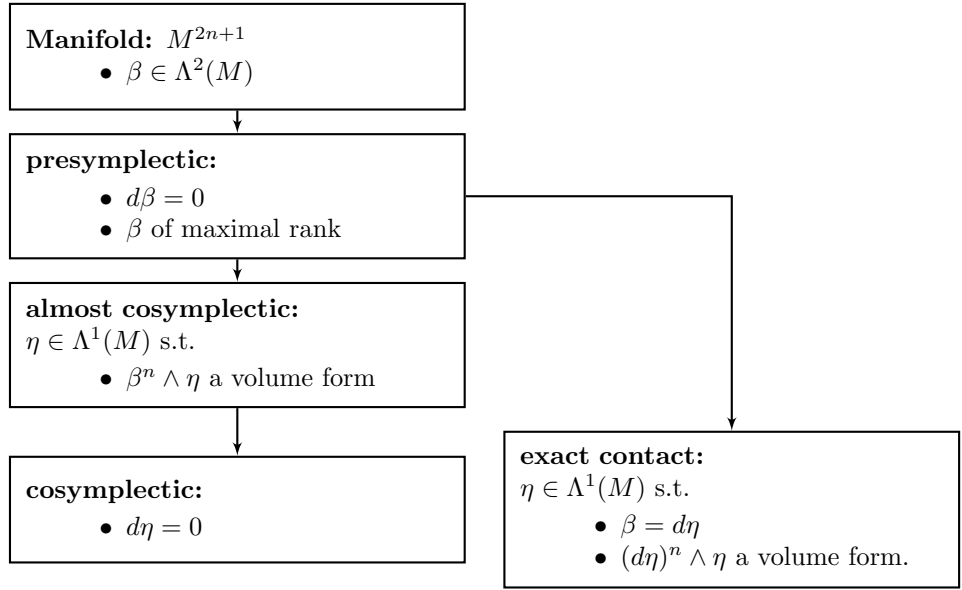
\captionof{figure}{Some geometrical structures on odd-dimensional manifolds. Here $\beta \in \Lambda^2(M)$ is a two-form on a manifold $M$ with $\dim{M} = 2n+1$.}\label{fig:structures}
\end{center}

\end{appendix}

\bibliographystyle{alpha}
\bibliography{presymplectic}

\begin{thebibliography}{EEMLRR99}

\bibitem[AM78]{Abraham78}
R.~Abraham and J.~Marsden.
\newblock {\em Foundations of Mechanics}.
\newblock Addison-Wesley, Redwood City, 1978.

\bibitem[Arn78]{Arnold78}
V.I. Arnold.
\newblock {\em Mathematical Methods of Classical Mechanics}.
\newblock Springer-Verlag, New York, 1978.
\newblock \url{https://www.springer.com/us/book/9780387968902}.

\bibitem[BKM20a]{Burby_Kallinikos_MacKay_2020}
J.~W. Burby, N.~Kallinikos, and R.~S. MacKay.
\newblock Generalized {G}rad–{S}hafranov equation for non-axisymmetric {MHD}
  equilibria.
\newblock {\em Physics of Plasmas}, 27(10):102504, 2020.
\newblock \url{https://doi.org/10.1063/5.0015420}.

\bibitem[BKM20b]{Burby_Kallinikos_MacKay_2020a}
J.~W. Burby, N.~Kallinikos, and R.~S. MacKay.
\newblock Some mathematics for quasisymmetry.
\newblock {\em J. Math. Phys.}, 61:093503, 2020.
\newblock \url{https://doi.org/10.1063/1.5142487}.

\bibitem[Boo81]{Boozer81}
A.~Boozer.
\newblock Plasma equilibrium with rational magnetic surfaces.
\newblock {\em Phys. Fluids}, 24:1999--2003, 1981.
\newblock \url{https://doi.org/10.1063/1.863297}.

\bibitem[dLS93]{leonCosymplecticReductionSingular1993}
M.~de~Leon and M.~Saralegi.
\newblock Cosymplectic reduction for singular momentum maps.
\newblock {\em J. Physics A}, 26(19):5033--5043, October 1993.
\newblock \url{https://doi.org/10.1088/0305-4470/26/19/032}.

\bibitem[dVN03]{deVediere00}
Y.C. de~Vediere and S.V. Ngoc.
\newblock Singular {B}ohr-{S}ommerfeld rules for {2D} integrable systems.
\newblock {\em Annales Scientifiques de l'École Normale Supérieure},
  36(1):1--55, 2003.
\newblock \url{https://doi.org/10.1016/S0012-9593(03)00002-8}.

\bibitem[EEMLRR99]{Enriquez99}
A.~Echeverria-Enriquez, M.C. Munoz-Lecanda, and N.~Roman-Roy.
\newblock Reduction of presymplectic manifolds with symmetry.
\newblock {\em Reviews in Mathematical Physics}, 11(10):1209--1247, 999.
\newblock \url{https://doi.org/10.1142/S0129055X99000386}.

\bibitem[Eli90]{Eliasson90}
L.H. Eliasson.
\newblock Normal forms for {H}amiltonian systems with {P}oisson commuting
  integrals--elliptic case.
\newblock {\em Commentarii Mathematici Helvetici}, 65:4--35, 1990.
\newblock \url{https://link.springer.com/article/10.1007/BF02566590}.

\bibitem[GB91a]{Garren_Boozer_91b}
D.~A. Garren and A.~H. Boozer.
\newblock Existence of quasihelically symmetric stellarators.
\newblock {\em Phys. Plasmas}, 3:2822, 1991.
\newblock \url{https://doi.org/10.1063/1.859916}.

\bibitem[GB91b]{Garren_Boozer_91a}
D.~A. Garren and A.~H. Boozer.
\newblock Magnetic field strength of toroidal plasma equilibria.
\newblock {\em Phys. Plasmas}, 3:2805, 1991.
\newblock \url{https://doi.org/10.1063/1.859915}.

\bibitem[GN79a]{Gotay79a}
M.J. Gotay and J.M. Nester.
\newblock Presymplectic {H}amilton and {L}agrange systems, gauge
  transformations and the {D}irac theory of constraints.
\newblock In W.~Beigelbock, A.~Bohm, and E.~Takasugi, editors, {\em Group
  Theoretical Methods in Physics}, volume~94 of {\em Lecture Notes in Physics},
  pages 272--279. Springer, 1979.
\newblock \url{https://link.springer.com/chapter/10.1007/3-540-09238-2_74}.

\bibitem[GN79b]{Gotay79b}
M.J. Gotay and J.M. Nester.
\newblock Presymplectic {L}agrangian systems. {I} : The constraint algorithm
  and the equivalence theorem.
\newblock {\em Annales de l'I.H.P. Physique théorique}, 30(2):129--142, 1979.
\newblock \url{http://www.numdam.org/item?id=AIHPA_1979__30_2_129_0}.

\bibitem[Got82]{gotayCoisotropicImbeddingsPresymplectic1982}
M.J. Gotay.
\newblock On coisotropic imbeddings of presymplectic manifolds.
\newblock {\em Proc. Amer. Math. Soc.}, 84:111--111, January 1982.
\newblock \url{https://doi.org/10.1090/S0002-9939-1982-0633290-X }.

\bibitem[Ham62]{Hamada62}
S.~Hamada.
\newblock Hydromagnetic equilibria and their proper coordinates.
\newblock {\em Nuclear Fusion}, 2(1-2):23--37, 1962.
\newblock \url{https://doi.org/10.1088/0029-5515/2/1-2/005}.

\bibitem[Hel14]{Helander14}
P.~Helander.
\newblock Theory of plasma confinement in non-axisymmetric magnetic fields.
\newblock {\em Rep. Prog. Phys.}, 77(8):087001, 2014.
\newblock \url{https://doi.org/10.1088/0034-4885/77/8/087001}.

\bibitem[HM03]{Hazeltine03}
R.D. Hazeltine and J.D. Meiss.
\newblock {\em Plasma Confinement}.
\newblock Dover Publications, Mineola, NY, 2nd edition, 2003.
\newblock \url{https://store.doverpublications.com/0486151034.html}.

\bibitem[Lit79]{Littlejohn_1979}
R.~L. Littlejohn.
\newblock A guiding center {H}amiltonian: A new approach.
\newblock {\em J. Math. Phys.}, 20:2445, 1979.
\newblock \url{https://doi.org/10.1063/1.524053}.

\bibitem[LS18]{Landerman_Sengupta_2018}
M.~Landerman and W.~Sengupta.
\newblock Direct construction of optimized stellarator shapes. part 1. theory
  in cylindrical coordinates.
\newblock {\em J. Plasma Phys.}, 84:905840616, 2018.
\newblock \url{https://doi-org.lanl.idm.oclc.org/10.1017/S0022377818001289}.

\bibitem[LS19]{Landerman_Sungupta_2019}
M.~Landerman and W.~Sengupta.
\newblock Constructing stellarators with quasisymmetry to higher order.
\newblock {\em J. Plasma Phys.}, 85:815850601, 2019.
\newblock \url{https://doi.org/10.1017/S0022377819000783}.

\bibitem[LSP19]{Landerman_Sengupta_Plunk_2019}
M.~Landerman, W.~Sengupta, and G.~G. Plunk.
\newblock Direct construction of optimized stellarator shapes. part 2.
  {N}umerical quasisymmetric solutions.
\newblock {\em J. Plasma Phys.}, 85:905850103, 2019.
\newblock \url{https://doi-org.lanl.idm.oclc.org/10.1017/S0022377818001344}.

\bibitem[Mac20]{MacKay20}
R.~S. MacKay.
\newblock Differential forms for plasma physics.
\newblock {\em J. Plas. Phys.}, 86(1):925860101, 2020.
\newblock \url{https://doi.org/10.1017/S0022377819000928}.

\bibitem[Mei17]{Meiss17a}
J.D. Meiss.
\newblock {\em Differential Dynamical Systems}, volume~22 of {\em Mathematical
  Modeling and Computation}.
\newblock SIAM, Philadelphia, revised edition, 2017.
\newblock \url{https://doi.org/10.1137/1.9781611974645}.

\bibitem[Min35]{mineur1935systemes}
H.~Mineur.
\newblock Sur les systemes m{\'e}caniques admettant $n$ int{\'e}grales
  premieres uniformes et l’extension a ces systemes de la m{\'e}thode de
  quantification de {S}ommerfeld.
\newblock {\em CR Acad. Sci., Paris}, 200:1571--1573, 1935.

\bibitem[OR04]{Ortega04}
J.-P. Ortega and T.S. Ratiu.
\newblock {\em Momentum Maps and {H}amiltonian Reduction}, volume 222 of {\em
  Progress in Mathematics}.
\newblock Birkhauser, Boston, MA, 2004.
\newblock \url{https://www.springer.com/gp/book/9780817643072}.

\bibitem[PLH19]{Plunk_Landerman_Helander_2019}
G.~G. Plunk, M.~Landerman, and P.~Helander.
\newblock Direct construction of optimized stellarator shapes. part 3.
  {O}mnigeneity near the magnetic axis.
\newblock {\em J. Plasma Phys.}, 85:905850602, 2019.
\newblock \url{https://doi.org/10.1017/S002237781900062X}.

\bibitem[SVN18]{sepeIntegrableSystemsSymmetries2018}
D.~Sepe and S.~V{\~u}~Ng{\d o}c.
\newblock Integrable systems, symmetries, and quantization.
\newblock {\em Lett. Math. Phys.}, 108(3):499--571, 2018.
\newblock \url{https://doi.org/10.1007/s11005-017-1018-z}.

\bibitem[Vey78]{veyCertainsSystemesDynamiques1978}
J.~Vey.
\newblock Sur {{Certains Systemes Dynamiques Separables}}.
\newblock {\em Amer. J. Math.}, 100(3):591--614, 1978.
\newblock \url{https://doi.org/10.2307/2373841}.

\bibitem[Wei77]{weinsteinLecturesSymplecticManifolds1977a}
A.~Weinstein.
\newblock {\em Lectures on Symplectic Manifolds}.
\newblock Number no. 29 in Regional Conference Series in Mathematics. {American
  Mathematical Society}, {Providence. R.I}, 1977.
\newblock \url{https://bookstore.ams.org/cbms-29/}.

\bibitem[Zun05]{Zung05}
N.T. Zung.
\newblock Convergence versus integrability in {B}irkhoff normal form.
\newblock {\em Ann. Math.}, 161(1):141--156, 2005.
\newblock \url{https://doi.org/10.4007/annals.2005.161.141}.

\bibitem[Zun18]{zungConceptualApproachProblem2018}
N.T. Zung.
\newblock A conceptual approach to the problem of action-angle variables.
\newblock {\em Arch Rational Mech Anal}, 229:789–833, 2018.
\newblock \url{https://doi.org/10.1007/s00205-018-1227-3}.

\end{thebibliography}

\end{document}